\documentclass[11pt]{article}

\usepackage{bbold}
\usepackage{amsmath,amsthm,amssymb,amsfonts}
\usepackage{colonequals}
\usepackage{fullpage}
\usepackage{dsfont}
\usepackage{graphicx}
\usepackage[table,dvipsnames]{xcolor}
\usepackage{hyperref}
\hypersetup{colorlinks=true,urlcolor=blue,citecolor=blue,linkcolor=blue}
\usepackage{url} 
\usepackage{collcell}
\usepackage{xstring}

\newcommand{\R}{\mathbb{R}}
\newcommand{\F}{\mathbb{F}}

\newcommand{\Span}{\mathrm{span}}

\newcommand{\zero}{\{0\}}
\newcommand{\X}{\mathcal{X}}
\newcommand{\V}{\mathcal{V}}

\newcommand{\U}{\mathcal{U}}
\newcommand{\EBasis}{E_{abcd}} 
\newcommand{\SpEBasis}{\Span\{\EBasis\}^\perp}

\newcommand{\SymTwoU}{S^2(\mathcal{U})}
\newcommand{\Sym}{\operatorname{Sym}}
\newcommand{\MEntries}{M_{abcd,ij}}
\newcommand{\MPEntries}{M''_{abcd,ij}}
\newcommand{\Z}{\mathbb{Z}}

\newtheorem{theorem}{Theorem}[section]
\newtheorem{theorem-star}[theorem]{Theorem*}
\newtheorem{remark}[theorem]{Remark}

\newtheorem{lemma}[theorem]{Lemma}

\newtheorem{definition}[theorem]{Definition}

\newtheorem{proposition}[theorem]{Proposition}
\newtheorem{corollary}[theorem]{Corollary}
\newtheorem{conjecture}[theorem]{Conjecture}
\newtheorem*{theorem*}{Theorem}

\definecolor{myblue}{RGB}{100, 143, 255}
\definecolor{mypurple}{RGB}{120, 94, 240}
\definecolor{mymagenta}{RGB}{220, 38, 127}
\definecolor{myorange}{RGB}{254, 97, 0}
\definecolor{myyellow}{RGB}{255, 176, 0}

\newcolumntype{C}{>{\collectcell\colorvalue}c<{\endcollectcell}}

\usepackage{soul}

\title{Improving the Threshold for Finding Rank-1 Matrices in a Subspace}

\usepackage{authblk}

\author[1]{Jeshu Dastidar}

\author[1]{Tait Weicht}

\author[1]{Alexander S.\ Wein}

\affil[1]{Department of Mathematics, University of California, Davis}

\date{}

\begin{document}

\maketitle

\begin{abstract}
We consider a basic computational task of finding $s$ planted rank-1 $m \times n$ matrices in a linear subspace $\U \subseteq \R^{m \times n}$ where $\dim(\U) = R \ge s$. The work of Johnston--Lovitz--Vijayaraghavan (FOCS 2023) gave a polynomial-time algorithm for this task and proved that it succeeds when ${R \le (1-o(1))mn/4}$, under minimal genericity assumptions on the input. Aiming to precisely characterize the performance of this algorithm, we improve the bound to ${R \le (1-o(1))mn/2}$ and also prove that the algorithm fails when ${R \ge (1+o(1))mn/\sqrt{2}}$. Numerical experiments indicate that the true breaking point is $R = (1+o(1))mn/\sqrt{2}$. Our work implies new algorithmic results for tensor decomposition, for instance, decomposing order-4 tensors with twice as many components as before.
\end{abstract}

\thispagestyle{empty}

\clearpage

\tableofcontents

\thispagestyle{empty}

\clearpage

\setcounter{page}{3}
\pagestyle{plain}

\section{Introduction} \label{Section:Introduction}

We consider the basic computational question of finding rank-1 matrices in a given linear subspace $\U \subseteq \R^{m \times n}$. We assume $\U$ is generated as the linear span of $s$ ``planted'' rank-1 matrices, along with $R-s$ other matrices, where $R = \dim(\U)$ and $0 \le s \le R$. As input, we receive an arbitrary basis for $\U$, and the goal is to output the $s$ planted matrices, up to permutation and scalar multiple (in the sense of Proposition~\ref{prop:extract-v}). Aside from being a fundamental question in its own right, the $s=R$ case is a key subroutine in various tensor decomposition algorithms, as we will discuss in Section~\ref{subsec:Motive&App}. In the $s=0$ case there are no planted matrices to recover, so here the goal is to algorithmically generate a \emph{certificate} that no rank-1 matrices exist in $\U$.

For a successful solution to this problem, we desire an algorithm that works for a wide range of possible inputs. To circumvent NP-hardness~\cite{Buss}, we follow~\cite{JLV} and assume that the matrices generating $\U$ are chosen \emph{generically}, either from the set of rank-1 matrices or the set of all matrices, as appropriate. The formal meaning of ``generic'' will be stated in Section~\ref{Section:Setting+MainResults}, but essentially this means we ask for an algorithm that works for ``almost all'' inputs, in the sense that the set of inputs where the algorithm fails has measure zero. We desire an algorithm that succeeds, in this sense, for the largest possible range of parameters $m,n,R,s$. Furthermore, we desire an algorithm with fast runtime.

The work of~\cite{JLV} gives a polynomial-time algorithm for the above problem, under the condition ${R \le \frac{1}{4}(m-1)(n-1)}$, and in fact for a more general setting where rank-1 matrices can be replaced by some other algebraic variety. We refer to this as the \emph{JLV algorithm}, and it will be a key focus of our work. The algorithm is based on the first level of the so-called \emph{Nullstellensatz hierarchy}. Higher levels of the hierarchy give rise to algorithms that are potentially more powerful at the expense of increased computational cost. Higher levels have been analyzed only in the fully non-planted case $s=0$, where they are shown to work even for nearly full-dimensional subspaces $R \le (1-\epsilon) mn$ with runtime $(mn)^{f(\epsilon)}$~\cite{DJL24} (see Remark~6.10 therein).

Without constraints on runtime, the \emph{identifiability} threshold is ${R = (m-1)(n-1)}$. By this we mean that for $0 \le s \le R \le (m-1)(n-1)$, the only rank-1 matrices in the span of $s$ generic rank-1 matrices along with $R-s$ generic matrices, will be the $s$ planted matrices (and their scalar multiples)~\cite[Theorem~4.6.14]{joins-book}. Furthermore, the planted matrices can be recovered by a sufficiently high level of the Nullstellensatz hierarchy, but the runtime will not be polynomial in general~\cite[Remark~4.6.16]{joins-book}. Conversely, if $R > (m-1)(n-1)$ then the subspace will contain ``spurious'' rank-1 matrices in addition to the planted ones (see~\cite{harris-book}, Definition~18.1(iii)). In fact, \emph{any} subspace of dimension $R \ge (m-1)(n-1)+2$ must contain infinitely many rank-1 matrices (see Definition~11.2 in~\cite{harris-book} or Theorem~3 in~\cite{entangled}).

\subsection{Our Contributions}
\label{subsec:Contribution}
In this work, our objective is to investigate the precise limits of the JLV algorithm, that is, the first Nullstellensatz level. Our main question is: for large $m,n$, what is the largest constant $c \in [1/4,1]$ such that the JLV algorithm succeeds on generic inputs with $R \le (c-o(1)) mn$? (We will see that there appears to be no dependence on $s$, to first order.) Even though higher levels of the hierarchy might tolerate larger $R$, the computational complexity quickly becomes impractical: the lowest level is $d=2$ and in general, any constant level $d \ge 2$ needs to compute with subspaces of dimension $\Theta((mn)^d)$. For this reason, we feel it is valuable to understand the exact performance of the first level, especially given the large gap between existing bounds: $1/4 \le c \le 1$.
Towards this, we have two main results which improve the bounds to $1/2 \le c \le 1/\sqrt{2}$. (In the special case $s=0$, the bound $c \ge 1-1/\sqrt{2} \approx 0.293$ was known previously~\cite[Theorem~6.9]{DJL24}.)

\begin{theorem*}[Informal, see Corollary~\ref{Cor:JLVSuccess}]
For any $0 \le s \le R$, if $R\le\frac{1}{2}(m-1)(n-2)$, the JLV algorithm succeeds on generic inputs.
\end{theorem*}

The asymmetry between $m$ and $n$ arises for technical reasons in the proof. The condition $R \le \frac{1}{2}(m-2)(n-1)$ is equally valid.

\begin{theorem*}[Informal, see Theorem~\ref{Thm:JLVFail}]
For any $0 \le s \le R$, if $R>\frac{1}{\sqrt{2}}mn+1$, the JLV algorithm fails on generic inputs.
\end{theorem*}

The precise meaning of ``fail'' here is that a particular key step in the analysis breaks down. We also give numerical experiments (Section~\ref{Section:Conj&Evidence}) suggesting that the true value is $c = 1/\sqrt{2}$, matching our second result. We highlight the open problem of proving this, which we state as Conjecture~\ref{Conjecture:SharpBoundary}. The improvement from $c = 1/4=0.25$ to, potentially, $c = 1/\sqrt{2} \approx 0.71$ represents a much wider range of values for $\dim(\U)$, as a fraction of the ambient dimension $mn$.

The negative result (Theorem~\ref{Thm:JLVFail}) follows from a straightforward dimension-counting argument, and our main contribution is the positive result (Corollary~\ref{Cor:JLVSuccess}). Here, our proof is rather different from the original analysis of JLV~\cite{JLV} in that we take a more concrete combinatorial approach. In contrast, JLV gives a more abstract argument that also works for other algebraic varieties.

\subsection{Motivation and Applications}
\label{subsec:Motive&App}

Aside from being a fundamental computational question in linear algebra, the problem of finding planted rank-1 matrices in a subspace is also important from a machine learning point of view, as it is a key subroutine in various tensor decomposition algorithms~\cite{JLV,KMW24OvercompleteTDKoszulYoung}. Tensor decomposition has many ML applications such as unsupervised learning, the method of moments, temporal data, multi-relational data, latent variables, phylogenetic trees, topic models, etc. See \cite{Kolda2009, rabanser2017introductiontensordecompositionsapplications, Moitra_2018, aravindan-chapter} for further discussion.

To illustrate our contributions to tensor decomposition, consider the (CP) decomposition problem for an order-3 tensor $T \in \R^{n_1 \times n_2 \times n_3}$: given a tensor with a rank-$R$ decomposition $T = \sum_{\ell=1}^R a^{\ell} \otimes b^{\ell} \otimes c^{\ell}$ where the components $a^{\ell},b^{\ell},c^{\ell}$ are generically chosen vectors from $\R^{n_1},\R^{n_2},\R^{n_3}$ respectively, and $(a \otimes b \otimes c)_{ijk} := a_i b_j c_k$, the goal is to recover the list of components up to inherent ambiguities (permutation of $[R]$ and scalar multiplication). The order-4 tensor decomposition problem is analogous but for $T \in \R^{n_1 \times n_2 \times n_3 \times n_4}$.

Building on~\cite{JLV}, our improved algorithmic result for finding rank-1 matrices implies improved results for poly-time tensor decomposition. Specifically, we can decompose generic order-4 $n \times n \times n \times n$ tensors up to rank $R \le \frac{1}{2} (n-2)^2$, allowing roughly twice as many components as the previous result $R \le \frac{1}{4} (n-1)^2$~\cite{JLV}. For order-3 $n_1 \times n_2 \times n_3$ tensors that are unbalanced in the sense $n_3 \ge n_1 n_2 / 2$, we similarly obtain a factor-of-2 improvement from $R \le \frac{1}{4}(n_1-1)(n_2-1)$ to $R \le \frac{1}{2}(n_1-1)(n_2-2)$. In Section~\ref{sec:Applications}, we present further details of these results and generalizations thereof.

\subsection{Related Work}

The problem we study was first defined by JLV~\cite{JLV}, but earlier precursors considered various special cases and variations. For instance, the case $s=R$ is closely tied with the celebrated \emph{FOOBI} algorithm for order-4 tensor decomposition~\cite{foobi-link,foobi}. In Section~III of \cite{foobi} it is claimed that this algorithm decomposes generic $n \times n \times n \times n$ tensors of rank up to $R \le \frac{1}{2} n(n+1)$. However, Appendix~A of \cite{JLV} points out an error in the proof and provides a correct proof under the weaker condition $R \le \frac{1}{4}(n-1)^2$. Our improvement to $R \le \frac{1}{2}(n-2)^2$ in Section~\ref{sec:Applications} recovers the original claim, to first order. Some extensions of FOOBI include~\cite{MSS,HSS-spectral}, but these do not improve the rank condition.

Another precursor to JLV is \cite{CP_Decomp2018}, which considers only the special case $s=1$ but generalizes it in a few directions: the linear subspace can be replaced by an affine space, and the set of rank-1 matrices can be replaced by rank-$r$ matrices or tensors.

JLV was first to consider the full range $0 \le s \le R$, and their result also allows the set of rank-1 matrices to be replaced by an arbitrary (conic) algebraic variety.

In the special case $s=0$, some extensions to JLV have appeared. First, \cite{bhaskara2024newtoolssmoothedanalysis} gives an algorithm in the case of \emph{smoothed} (rather than generic) inputs, which amounts to giving a detailed analysis of noise-robustness. Also, \cite{DJL24} analyzes higher levels of the Nullstellensatz hierarchy, showing that they succeed for neary full-dimensional subspaces. Generalizing these results to $s > 0$ remains an interesting open problem.

Our improvements to order-3 tensor decomposition apply to unbalanced tensors, e.g.\ of shape roughly $n \times n \times n^2$. For ``square'' $n \times n \times n$ tensors, the classical \emph{simultaneous diagonalization} method can decompose generic tensors of rank up to $R \le n$~\cite{LRA-3-way}. Recent advances \cite{persu-thesis,koiran} have culminated in an algorithm approaching rank $2n$ \cite{KMW24OvercompleteTDKoszulYoung}. This method uses JLV as a subroutine, but the factor $1/4$ is not the bottleneck in the analysis, so our result does not imply an improvement here.

Other algebraic methods for tensor decomposition were proposed prior to JLV, including~\cite{cpd-3rd,normal-form} and references therein. To our knowledge, these have not been proven to succeed on generic inputs.

\section*{Acknowledgments}
\addcontentsline{toc}{section}{Acknowledgments}

We are thankful to Ben Lovitz and Aravindan Vijayaraghavan for helpful discussions. In particular, Ben Lovitz pointed us to many relevant references, e.g., concerning the identifiability question.

{\bf Funding:} All authors were supported by NSF CAREER Award CCF-2338091 (to A.~Wein). ASW was additionally supported by a Sloan Research Fellowship.

\section{Setting and Main Results}
\label{Section:Setting+MainResults}

Let us now formalize the basic computational question of interest. Throughout this paper, $\V=\R^{m\times n}$ is a vector space over the field $\R$ and $\U\subseteq\V$ is an $R$-dimensional subspace. We always assume $m,n \ge 2$. A \emph{rank-1} (technically, rank-at-most-1) matrix $v \in \R^{m \times n}$ is one that can be expressed as $v = xy^\top$ for vectors $x \in \R^m, y \in \R^n$. Equivalently, a matrix has rank 1 if and only if all its $2 \times 2$ minors vanish, so the set of rank-1 matrices can be expressed as the algebraic variety
$$\X=\{v\in\R^{m\times n}:f(v)=0 \text{ for all }f\in P\}\subseteq\V,$$
where
\begin{align*} \label{PolynSet}
P=\{v_{ab}v_{cd} - v_{ad}v_{cb} : 1\leq a<c\leq m, 1\leq b<d\leq n\} \end{align*}
with $P\subseteq\R[v_{ij}:1\leq i \leq m, 1\leq j\leq n]$. We assume $\U = \Span\{v^1,\ldots v^s,v^{s+1},\ldots,v^R\}$, where $v^1,\ldots,v^s\in\X$ are the \textit{planted} matrices and $v^{s+1},\ldots,v^R\in\V$ are the \textit{non-planted} matrices. The planted ones ($i \le s$) can be expressed as $v^i=x^i(y^i)^\top$ for vectors $x^i \in \R^m, y^i \in \R^n$. We view $x,y$ as symbolic variables, i.e., there is a variable for each entry of each vector $x^i$ or $y^i$, so $s(m+n)$ variables in total. For the non-planted ones ($i > s$), we write $v^i = z^i$ for symbolic variables $z$, which is $(R-s)mn$ variables in total. We will assume the underlying variables $x,y,z$ are chosen \emph{generically}, in the following sense.

\begin{definition}
\label{def:genericity}
For symbolic variables $w=(w_1,\ldots,w_k)$, a property holds ``generically'' or ``for generically chosen $w$'' if there exists a nonidentically-zero polynomial $p(w_1,\ldots,w_k) \in \R[w]$ such that the property holds true for all $w\in\R^k$ such that $p(w)\neq0$.
\end{definition}

This is a standard notion from algebraic geometry and it is well known that a generic property holds for all but a measure-zero set of points $w$. In our case, $w$ is the concatenation of all variables $x,y,z$, and we refer to
\[ \U = \Span\{x^1(y^1)^\top,\ldots,x^s(y^s)^\top,z^{s+1},\ldots,z^R\} \subseteq \R^{m\times n} \]
as a ``generically chosen planted subspace with parameters $0 \le s \le R$.'' Generically, such a subspace has dimension exactly $R$, as we see below.

\begin{lemma}
\label{lem:lin-indep}
If $\text{ }\U = \Span\{v^1,\ldots,v^R\} \subseteq\R^{m\times n}$ is a generically chosen planted subspace with parameters $0 \le s \le R$ where $R \le mn$, then the elements $v^1,\ldots,v^R$ are linearly independent.
\end{lemma}

This fact appears in \cite{JLV} and holds more generally for an irreducible algebraic variety in place of the rank-1 matrices. We give a self-contained proof for completeness.

\begin{proof}
Let $A$ be the $mn\times R$ matrix formed with columns $v^1,\ldots,v^R$. By assumption, the underlying variables $x,y,z$ are chosen generically. So it suffices to show that there is a square submatrix of size $R \times R$ in $A$ with determinant a non-zero polynomial in $x,y,z$. Towards this, it is enough to demonstrate a choice of vectors $x^{\ell},y^{\ell},z^{\ell}$ such that $\{x^\ell(y^\ell)^\top = x^{\ell}\otimes y^{\ell} : 1 \leq \ell \leq s\} \cup \{z^{\ell} : s < \ell \leq R\}$ is linearly independent. Fix an injection $\psi: [R] \rightarrow [m] \times [n]$. For each $\ell \in [R]$, define the following rule: if $\ell \mapsto (i,j)$, then set $x^\ell = e^i$, $y^\ell = e^j$, and $z^\ell = e^{ij}$, where $e^i$ is the $i$-th standard basis vector and $e^{ij}$ is the vectorization of the matrix with 1 in $(i,j)$ entry and 0 elsewhere. Now the submatrix of $A$ with rows indexed by the image of $\psi$ becomes the identity matrix, and hence has non-zero determinant.
\end{proof}

Before stating our main results, we need to unpack some details of the JLV algorithm. Toward this, for vectors $u,u' \in \V = \R^{m \times n}$, the associated \textit{symmetric 2-tensor} is \[u\vee u':=\frac{1}{2}u\otimes u' + \frac{1}{2}u'\otimes u.\]
For a subspace $U \subseteq \V$, we denote the space of such tensors by $S^2(U)=\Span\{u\vee u':u,u'\in U\}$. If $U$ has basis $\{u^1,\ldots,u^k\}$, then $S^2(U)$ has basis $\{u^i \vee u^j : 1 \le i \le j \le k\}$ and thus has dimension $\binom{k+1}{2}$.

Given a homogeneous degree-2 polynomial $f$, there exists a unique symmetric tensor $T^f \in S^2(\V)$ such that $f(v) =\langle v^{\otimes 2}, T^f \rangle$. Here we use the usual inner product $\langle T,S \rangle = \sum_{ij} T_{ij} S_{ij}$. We now write these tensors $T^f$ explicitly for the $2 \times 2$ minors $f \in P$. For $1\leq a<c\leq m$ and $1\leq b<d\leq n$, define $\EBasis \in S^2(\V)$ by
\[E_{abcd}=e_{ab}\vee e_{cd}-e_{ad}\vee e_{cb},\] where $e_{\alpha\beta}$ is the matrix with $1$ in the $(\alpha,\beta)$ entry and $0$ elsewhere. This is designed so that $E_{abcd} = T^f$ for the $2 \times 2$ minor $f(v) = v_{ab}v_{cd} - v_{ad}v_{cd}$, or in other words, $\langle v^{\otimes 2}, E_{abcd} \rangle = v_{ab} v_{cd}-v_{ad} v_{cb}$. To verify this, recall the identity $\langle a\otimes b,c\otimes d \rangle=\langle a,c \rangle\langle b,d \rangle$, so
\begin{align*}
\langle v^{\otimes 2}, E_{abcd} \rangle &= \langle v\otimes v, e_{ab}\vee e_{cd}-e_{ad}\vee e_{cb} \rangle \\
&= \langle v\otimes v, e_{ab}\vee e_{cd} \rangle - \langle v\otimes v, e_{ad}\vee e_{cb} \rangle \\
&= \langle v, e_{ab} \rangle\langle v, e_{cd} \rangle - \langle v, e_{ad} \rangle\langle v, e_{cb} \rangle \\
&= v_{ab} v_{cd}-v_{ad} v_{cb}.
\end{align*} 
Notice that $\langle v^{\otimes 2}, E_{abcd} \rangle=0$ if $v\in\X$.

A crucial step in the analysis of JLV is to establish
\[S^2(\U)\ \cap\ \SpEBasis = \Span\{(v^1)^{\otimes 2},\ldots,(v^s)^{\otimes 2}\},
\label{eq:goal} \tag{Goal}\]
where $(a,b,c,d)$ ranges over $1 \le a < c \le m$ and $1 \le b < d \le n$. Observe that we have the containment ``$\supseteq$" for free since $(v^i)^{\otimes 2}=v^i\vee v^i \in S^2(\U)$ and for $i \le s$, $\langle (v^i)^{\otimes 2},E_{abcd}\rangle=0$ since $v^1,\ldots,v^s\in\X$, so $(v^i)^{\otimes 2}\in\SpEBasis$. Establishing the reverse inclusion is the main challenge.
The reason \eqref{eq:goal} is helpful is because the left-hand side is a subspace that the algorithm can compute, which gives us access to the right-hand side. Given the right-hand side, JLV show how to extract the desired solutions $v^1,\ldots,v^s$, up to permutation and scalar multiple. This is captured by the following fact.

\begin{proposition}[\cite{JLV}] \label{prop:extract-v}
Given $\Span\{(v^1)^{\otimes 2},\ldots,(v^s)^{\otimes 2}\}$ with $v^1,\ldots,v^s$ linearly independent, there is a poly-time algorithm to output $\hat{v}^1,\ldots,\hat{v}^s$ for which there exists a permutation $\pi$ of $[s]$ and non-zero scalars $\alpha_1,\ldots,\alpha_s$ with $\hat{v}^i = \alpha_i v^{\pi(i)}$ for all $i \in [s]$.
\end{proposition}

Given a basis $\{P_1,\ldots,P_s\}$ for $\Span\{(v^1)^{\otimes 2},\ldots,(v^s)^{\otimes 2}\}$, the procedure for extracting $v^1,\ldots,v^s$ is Step~2 of Algorithm~1 in \cite{JLV}. The procedure involves decomposing a particular tensor using simultaneous diagonalization, and its correctness is justified in the proof of Corollary~16 in \cite{JLV}. See also Lemma~3.4 of \cite{KMW24OvercompleteTDKoszulYoung}.

In the case $s=0$, if \eqref{eq:goal} holds, i.e., if the left-hand side is $\{0\}$, this constitutes a certificate that $\U$ contains no rank-1 matrices. This is captured by the following fact.

\begin{proposition}[\cite{JLV}, Observation 11]
\label{prop:Obv11}
If $\SymTwoU \cap \SpEBasis = \zero$ for $1\leq a<c \leq m$ and $1\leq b<d \leq n$, then $\U\cap\X = \zero$.
\end{proposition}

The proof is immediate from our discussion above: if there were a non-zero element $v \in \U \cap \X$ then $v^{\otimes 2}$ would belong to both $S^2(\U)$ and $\SpEBasis$.

We have now essentially described the JLV algorithm: given a basis $\{u^1,\ldots,u^R\}$ for $\U$, it first computes the left-hand side of \eqref{eq:goal} --- making use of the basis $\{u^i \vee u^j : 1 \le i \le j \le R\}$ for $S^2(\U)$ --- and then applies either Proposition~\ref{prop:extract-v} ($s > 0$) or Proposition~\ref{prop:Obv11} ($s=0$) as appropriate. Our main results give conditions under which \eqref{eq:goal} holds or fails to hold.

\begin{theorem} \label{Thm:GoalHolds}
If $\text{ }\U = \Span\{v^1,\ldots,v^R\} \subseteq\R^{m\times n}$ is a generically chosen planted subspace with parameters $0 \le s \le R$ where $R\leq\frac{1}{2}(m-1)(n-2)$, then \eqref{eq:goal} holds.
\end{theorem}

Combining this with Lemma~\ref{lem:lin-indep} and Proposition~\ref{prop:extract-v}, this immediately yields our positive result for the JLV algorithm.

\begin{corollary} \label{Cor:JLVSuccess}
If $\text{ }\U = \Span\{v^1,\ldots,v^R\} \subseteq\R^{m\times n}$ is a generically chosen planted subspace with parameters $0 \le s \le R$ where $R\leq\frac{1}{2}(m-1)(n-2)$, then the JLV algorithm succeeds in recovering $v^1,\ldots,v^s$, up to permutation and scalar multiple.
\end{corollary}

We also have the following negative result for \eqref{eq:goal}.

\begin{theorem} \label{Thm:JLVFail}
If $\U = \Span\{v^1,\ldots,v^R\} \subseteq\R^{m\times n}$ is a generically chosen planted subspace with parameters $0 \le s \le R$ where $\binom{m}{2}\binom{n}{2}<\binom{R+1}{2} - s$, then \eqref{eq:goal} fails to hold.
\end{theorem}

Since $s \le R$, the condition $\binom{m}{2}\binom{n}{2}<\binom{R+1}{2} - s$ is implied by
\[ \frac{m^2}{2} \cdot \frac{n^2}{2} < \frac{(R-1)^2}{2}, \]
which is implied by the simpler condition $R>\frac{1}{\sqrt{2}}mn+1$ claimed in Section~\ref{subsec:Contribution}.

When \eqref{eq:goal} fails, we consider this ``failure'' of the JLV algorithm, in the sense that a key step in the analysis breaks down. Numerical experiments in Section~\ref{subsec:NumericalExperiments} reveal that, unsurprisingly, JLV indeed fails to output the planted rank-1 matrices in these cases. This is essentially because it attempts to run a tensor decomposition algorithm on a tensor that has no reason to have low enough rank.

\subsection*{Organization}
\addcontentsline{toc}{subsection}{Organization}

We leave the reader with a road map. The negative result (Theorem~\ref{Thm:JLVFail}) follows from a straightforward dimension-counting argument, and the proof is deferred to Section~\ref{subsec:NegResult}. Most of the work is for the positive result (Theorem~\ref{Thm:GoalHolds}), and we outline the ideas in Sections~\ref{Section:Prelims&Perspective}--\ref{sec:MatrixAndProperties} with the culminating proof in Section~\ref{subsec:PosResult}. Numerical results indicating that the negative result is tight are in Section~\ref{Section:Conj&Evidence}. Finally, Section~\ref{sec:Applications} details the implications for tensor decomposition.

\section{Preliminaries and Perspective}
\label{Section:Prelims&Perspective}

Towards proving our \eqref{eq:goal}, we reformulate it via linear algebra and set up the necessary ingredients in the proof of Theorem~\ref{Thm:GoalHolds}.

\subsection{A Linear Algebraic View}
\label{subsec:LinAlgApproach}

For the containment ``$\subseteq$" in \eqref{eq:goal}, we first define a matrix $M$ where the columns are indexed by
\[(i,j)\in\Omega:=\{(i,j) : 1\leq i<j\leq R \text{ or } s < i = j\leq R\}\]
and the rows are indexed by tuples $(a,b,c,d)$ with $1\leq a<c\leq m$, $1\leq b<d\leq n$. We define the entries of $M$ by
\[ \MEntries := \langle v^i\vee v^j,E_{abcd}\rangle. \]
We will explicitly compute these entries soon. We now have the following lemma that gives a concrete linear algebraic approach to prove \eqref{eq:goal}.

\begin{lemma} \label{Lem:GoalLemma}
If $M$ has full column rank, then \eqref{eq:goal} holds.
\end{lemma}

While not needed for our proof, we note that the converse also holds, i.e., \eqref{eq:goal} is equivalent to $M$ having full column rank, as long as $v^1,\ldots,v^R$ are linearly independent.

\begin{proof} 
We prove this via contrapositive. Suppose \eqref{eq:goal} is not met, meaning
\[\Span\{v^i\vee v^j\}\ \cap\ \SpEBasis \not\subseteq \Span\{(v^1)^{\otimes 2},\ldots,(v^s)^{\otimes 2}\}.\]
Then there exists an $\alpha=(\alpha_{ij})_{1\leq i\leq j\leq R}$ such that
\begin{equation}\label{eqn:ContainmentEqn}
\sum_{1\leq i\leq j\leq R}\alpha_{ij}v^i\vee v^j \in \SpEBasis \setminus \Span\{(v^i)^{\otimes 2} : 1\leq i\leq s\}. 
\end{equation}
Observe that this implies there exists a pair $(i,j)\in\Omega$ such that $\alpha_{ij}\neq0$. If not, we have $\alpha_{ij}=0$ for all $(i,j)\in\Omega$. So,
\begin{align*}
\sum_{1\leq i\leq j\leq R}\alpha_{ij} v^i\vee v^j &= \sum_{(i,j)\in\Omega}\alpha_{ij} v^i\vee v^j + \sum_{(i,j)\notin\Omega}\alpha_{ij} v^i\vee v^j \\
&= \sum_{(i,j)\in\Omega} 0 \cdot v^i\vee v^j + \sum_{\{(i,j) \,:\, 1\leq i=j\leq s\}} \alpha_{ij} v^i\vee v^j \\
&= \sum_{1\leq i \leq s}\alpha_{ii} v^i\vee v^i \\
&\in \Span\{(v^i)^{\otimes 2}:1\leq i\leq s\},
\end{align*}
where the last line follows from $v^i\vee v^i = (v^i)^{\otimes 2}$. This contradicts \eqref{eqn:ContainmentEqn}, thus the observation must hold.
Now, for all $1\leq a<c\leq m$ and $1\leq b<d\leq n$, we have
\begin{align*}
\hspace{1.5cm} 0 &= \left\langle \sum_{1\leq i\leq j\leq R}\alpha_{ij} v^i\vee v^j, \EBasis \right\rangle \\
&= \sum_{1\leq i \leq s}\alpha_{ii} \langle v^i\vee v^i, \EBasis \rangle + \sum_{(i,j)\in\Omega}\alpha_{ij} \langle v^i\vee v^j, \EBasis \rangle \\
&= \sum_{(i,j)\in\Omega}\alpha_{ij} \langle v^i\vee v^j, \EBasis \rangle \\
&= \sum_{(i,j)\in\Omega}\alpha_{ij} \MEntries,
\end{align*} 
where the third equality follows since $\langle v^i \vee v^i, \EBasis \rangle = 0$ for $v^1,\ldots,v^s\in\X$.
This shows that there exists a non-zero $\alpha=(\alpha_{ij})_{(i,j)\in\Omega}$ such that $M\alpha=0$. In other words, the columns of $M$ are not linearly independent.
\end{proof}

The following general fact will be helpful for computing the entries of $M$.
\begin{lemma}
$\langle a\vee b, c\vee d\rangle = \frac{1}{2}\left(\langle a,c\rangle\langle b,d\rangle + \langle a, d\rangle\langle b,c\rangle\right)$.
\end{lemma}
\begin{proof}
We compute
\begin{align*}
\langle a\vee b, c\vee d\rangle &= \frac{1}{4}\langle a\otimes b+b\otimes a,c\otimes d+d\otimes c\rangle \\
&=\frac{1}{4}\Big(\langle a\otimes b,c\otimes d\rangle + \langle a\otimes b,d\otimes c\rangle + \langle b\otimes a,c\otimes d\rangle + \langle b\otimes a,d\otimes c\rangle \Big) \\
&= \frac{1}{4}\Big( \langle a,c \rangle\langle b,d \rangle + \langle a,d \rangle\langle b,c \rangle + \langle b,c \rangle\langle a,d \rangle + \langle b,d \rangle\langle a,c \rangle\Big) \\
&= \frac{1}{2}\Big(\langle a,c\rangle\langle b,d\rangle + \langle a, d\rangle\langle b,c\rangle\Big),
\end{align*}
where the third equality uses the identity $\langle a\otimes b,c\otimes d \rangle=\langle a,c \rangle\langle b,d \rangle$.
\end{proof}
Now, we are in a position to compute the entries $\MEntries$ of $M$. 
\begin{lemma}\label{Lem:EntryEqn}
The entries of $M$ are given by
\[ \MEntries=\frac{1}{2}\left(v_{ab}^iv_{cd}^j+v_{cd}^iv_{ab}^j-v_{ad}^iv_{cb}^j-v_{cb}^iv_{ad}^j\right) \]
for $1\leq a<c\leq m$, $1\leq b<d\leq n$, and $(i,j)\in\Omega$.
\end{lemma}
\begin{proof}
We compute
\begin{align*}
\MEntries &= \langle v^i\vee v^j,E_{abcd}\rangle \\
&= \langle v^i\vee v^j,e_{ab}\vee e_{cd}-e_{ad}\vee e_{cb} \rangle \\
&= \langle v^i\vee v^j,e_{ab}\vee e_{cd}\rangle-\langle v^i\vee v^j,e_{ad}\vee e_{cb} \rangle \\
&= \frac{1}{2}\Big(\langle v^i,e_{ab}\rangle\langle v^j,e_{cd}\rangle + \langle v^i,e_{cd}\rangle\langle v^j,e_{ab}\rangle - \langle v^i,e_{ad}\rangle\langle v^j,e_{cb}\rangle - \langle v^i,e_{cb}\rangle\langle v^j,e_{ad}\rangle\Big) \\
&= \frac{1}{2}\left(v_{ab}^iv_{cd}^j + v_{cd}^iv_{ab}^j - v_{ad}^iv_{cb}^j - v_{cb}^iv_{ad}^j\right),
\end{align*}
as desired.
\end{proof}

We have now expressed the entries of $M$ in terms of the $v$ variables. Recall that the $v$ variables can in turn be expressed in terms of the underlying $x,y,z$ variables, namely:
\[v^\ell_{\alpha\beta} := \begin{cases}
x^\ell_\alpha y^\ell_\beta & \text{if } 1\leq\ell\leq s, \\
z^\ell_{\alpha\beta} & \text{if } s<\ell\leq R.
\end{cases}\]
This allows us to view $M$ as a matrix of symbolic polynomials in the variables $x,y,z$.

\subsection{Proof Strategy}
\label{Section:ProofStrategy}

Let us now discuss the high-level strategy towards \eqref{eq:goal} and set the stage for the proof. In Section \ref{Section:Construction}, we will construct a square submatrix $M'$ of $M$ that includes all columns of $M$ (and a subset of the rows). We think of $M' = M'(x,y,z)$ as a symbolic matrix in the variables $x,y,z$. Our goal will be to show that $\det(M')$, viewed as a polynomial in the variables $x,y,z$, is not identically zero. By Lemma \ref{Lem:GoalLemma}, this will imply Theorem \ref{Thm:GoalHolds}.

It suffices to show a stronger statement: $\det(M') \not\equiv 0$ even after substituting $0$ for certain $x,y,z$ variables. We will next specify which variables to zero out and which to keep.

Let $k\in\Z_+$ be the largest integer such that $2k+1\leq n$, which means $2(k+1)+1 > n$, i.e., $n \le 2(k+1)$. Fix an arbitrary injection \[\varphi:[R]\longrightarrow[m-1]\times[k]\quad\text{given by}\quad i\longmapsto(f(i),g(i)).\] Note that such an injection exists, since $R \le (m-1)k$, which is true by our assumption $R \le \frac{1}{2}(m-1)(n-2)$ combined with $n \le 2(k+1)$.

For intuition, the above injection lets us associate each $i \in [R]$ with a position $(f(i),g(i))$ in the top-left $(m-1) \times k$ block of an $m \times n$ matrix. Row $m$ of the matrix is not included because it will play a special row. Similarly, column $k+1$ will play a special role, and columns $k+2,\ldots,2k+1$ will serve as ``duplicates'' of columns $1,\ldots,k$, with $b \in [k]$ paired with $b+k+1$.

We will zero out the following variables and keep the rest as formal variables.
\begin{itemize}
\item Set $x^\ell_{\alpha}=0$ unless $\alpha\in\{f(\ell),m\}$.
\item Set $y^\ell_{\beta}=0$ unless $\beta\in\{g(\ell),k+1,g(\ell)+k+1\}$.
\item Set $z^\ell_{\alpha\beta}=0$ unless both of the above conditions hold: $\alpha\in\{f(\ell),m\}$ and $\beta\in\{g(\ell),k+1,g(\ell)+k+1\}$.
\end{itemize}
Note that $x^\ell_{\alpha}y^\ell_{\beta} = 0$ unless $\alpha\in\{f(\ell),m\}$ and $\beta\in\{g(\ell),k+1,g(\ell)+k+1\}$. So, the above rules for $x,y,z$ yield rules for the $v$ variables as follows:
\begin{equation}\label{eq:v-rule}
v^\ell_{\alpha\beta}=0 \text{ unless}
\begin{cases}
\alpha \in \{f(\ell),m\} \text{ and} \\
\beta\in\{g(\ell), k+1, g(\ell)+k+1\}.
\end{cases}
\end{equation}
These rules will guide the construction of the submatrix $M'$ in the next section.

\section{Properties of two Important Matrices}
\label{sec:MatrixAndProperties}

\subsection{Construction of $M'$ and $M''$} \label{Section:Construction}

We will construct the square submatrix $M'$ of $M$ as follows. Recall that $M'$ will include all columns of $M$, which are indexed by $\Omega$. It will be helpful to partition these columns into ``bins'' based on the following conditions on ordered pairs $(i,j)\in\Omega$, which involve $f,g$ as defined in Section \ref{Section:ProofStrategy}. Figure~\ref{fig:Bin1-4} provides a visualization of these bins.
\begin{itemize}
    \item \textbf{Bin 1} is the case $s< i=j\leq R$. So $f(i) = f(j)$ and $g(i)=g(j)$.
    \item \textbf{Bin 2} is the case $1\leq i<j\leq R$, $f(i) = f (j)$, and $g(i) \neq g(j)$.
    \item \textbf{Bin 3} is the case $1\leq i<j\leq R$, $f(i) \neq f (j)$, and $g(i) = g(j)$.
    \item \textbf{Bin 4} is the case $1\leq i<j\leq R$, $f(i) \neq f (j)$, and $g(i) \neq g(j)$.
\end{itemize}

\begin{figure}[h]
    \centering
    \includegraphics[width=0.3\linewidth]{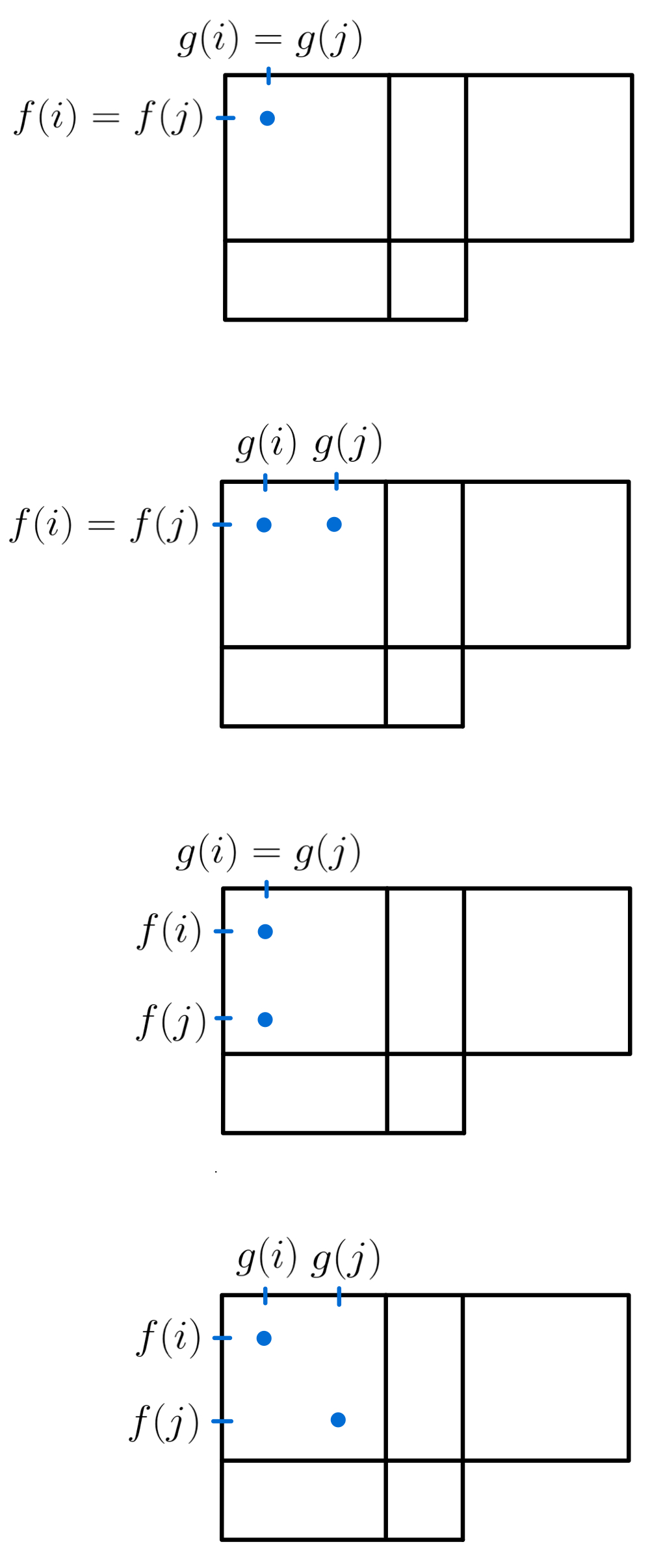}
    \caption{A visualization of Bins 1--4, from top to bottom. An $(i,j)$ pair is visualized as two positions, namely $(f(i),g(i))$ and $(f(j),g(j))$, in the top-left $(m-1) \times k$ block of an $m \times n$ matrix. These two positions are drawn as blue dots. The figures above are representative examples of the various cases that can occur, e.g., $f(i) < f(j)$ vs $f(i) > f(j)$.}
    \label{fig:Bin1-4}
\end{figure}

Recall that the rows of $M$ are indexed by $(a,b,c,d)$ with $1\leq a<c\leq m$ and $1\leq b<d\leq n$. We now select which of these rows to include in $M'$. These rows will be partitioned into ``classes'' which have some correspondence to the ``bins'' above. Figure~\ref{fig:ClassI-V} provides a visualization of these classes.

\begin{itemize}

\item \textbf{Class I.} For each $(i,j)$ in Bin 1, include the row $(a,b,c,d)$ where $\{a,c\}=\{f(i),m\}$ and $\{b,d\}=\{g(i),k+1\}$. (There is a unique such row, since $a<c$ and $b<d$ are required.)

\item \textbf{Class II.} For each $(i,j)$ in Bin 2, include the row $(a,b,c,d)$ where $\{a,c\}=\{f(i),m\}$ and $\{b,d\}=\{g(i),g(j)\}$.

\item \textbf{Class III.} For each $(i,j)$ in Bin 3, include the row $(a,b,c,d)$ where $\{a,c\}=\{f(i),f(j)\}$ and $\{b,d\}=\{g(i),k+1\}$. 

\item \textbf{Class IV.} For each $(i,j)$ in Bin 4, include the row $(a,b,c,d)$ where $\{a,c\}=\{f(i),f(j)\}$ and $\{b,d\}=\{g(i),g(j)\}$.

\item \textbf{Class V.} For each $(i,j)$ in Bin 4, include the row $(a,b,c,d)$ where $\{a,c\}=\{f(i),f(j)\}$ and $\{b,d\}=\{g(i)+k+1,g(j)+k+1\}$.

\end{itemize}

\begin{figure}
    \centering
    \includegraphics[width=0.28\linewidth]{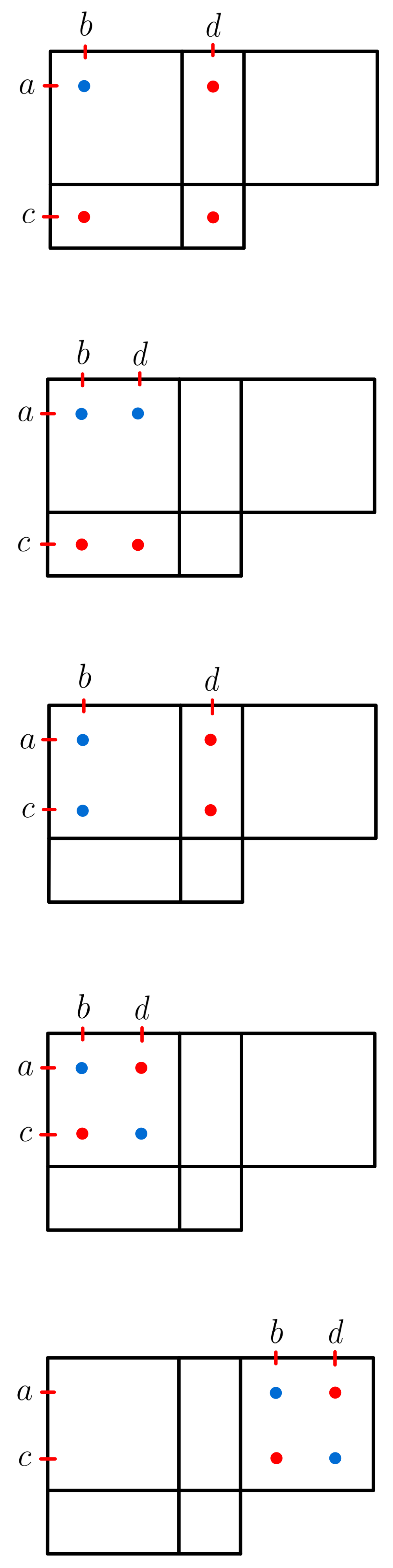}
    \caption{A visualization of Classes I--V, from top to bottom. A tuple $(a,b,c,d)$ is visualized as a $2 \times 2$ submatrix of an $m \times n$ matrix, namely the one with rows $\{a,c\}$ and columns $\{b,d\}$. The 4 positions in the submatrix are drawn as dots, with the blue dots indicating overlap with the corresponding $(i,j)$, visualized as in Figure~\ref{fig:Bin1-4}.}
    \label{fig:ClassI-V}
\end{figure}

We need to justify that $M'$ is a square matrix. Each column $(i,j)$ in Bins 1--3 corresponds uniquely to one of the rows in Classes I--III. Each $(i,j)$ in Bin 4 is paired with another $(i',j')$ in Bin 4 such that $\{f(i),f(j)\} = \{f(i'),f(j')\}$ and $\{g(i),g(j)\} = \{g(i'),g(j')\}$. This pair of columns corresponds uniquely to a pair of rows, one in Class IV and the other in Class V. This shows that $M'$ is square as desired.

\begin{definition}
Define $\Lambda$ to be the set of $(a,b,c,d)$ pairs indexing the rows of $M'$, i.e., those specified in Classes I--V above.
\end{definition}

\begin{definition} \label{Defn:M''}
Define $M''(x,y,z)$ to be the symbolic matrix obtained from $M'(x,y,z)$ by substituting 0 for certain $x,y,z$ variables as described in Section \ref{Section:ProofStrategy}, and leaving the rest as formal variables.
\end{definition}

\subsection{Zero Lemma}
\label{Section:ZeroLemma}

Towards our goal of showing $\det(M'') \not\equiv 0$, we aim to show that $M''$ is, in a sense, essentially upper-triangular. We will need to argue that many entries of $M''$ are identically zero (as polynomials in $x,y,z$), and the following ``Zero Lemma'' will be a useful tool for this.

\begin{lemma}[\textbf{Zero Lemma}] \label{Lem:GenLemma}
Let $(i,j)\in\Omega$ and $(a,b,c,d)\in\Lambda$. Denote \[\mathcal{R} = \{f(i),f(j),m\}\] and \[\mathcal{C} = \{g(i),g(j),k+1,g(i)+k+1,g(j)+k+1\}.\]
If $\{a,c\}\nsubseteq \mathcal{R}$ or $\{b,d\}\nsubseteq \mathcal{C}$, then $\MPEntries=0$.
\end{lemma}

\begin{proof}
Suppose the given assumption. We have two cases.

\paragraph{Case 1.} Suppose $\{a,c\}\nsubseteq \mathcal{R}$. Then either $a\notin \mathcal{R}$ or $c\notin \mathcal{R}$.

\begin{itemize} 
\item If $a \notin \mathcal{R}$, then the rule in Eq.~\eqref{eq:v-rule} implies $v^\ell_{a\beta} = 0$ whenever $\ell\in\{i,j\}$, for any $\beta$. So
\begin{align*}
\MPEntries &= v^i_{ab}v^j_{cd} + v^i_{cd}v^j_{ab} - v^i_{ad}v^j_{cb} - v^i_{cb}v^j_{ad} \\
&= 0\cdot v^j_{cd} + v^i_{cd}\cdot0 - 0\cdot v^j_{cb} - v^i_{cb}\cdot0 \\
&= 0.
\end{align*}

\item If $c \notin \mathcal{R}$, then the rule in Eq.~\eqref{eq:v-rule} implies $v^\ell_{c\beta}=0$ whenever $\ell\in\{i,j\}$, for any $\beta$. So $\MPEntries=0$ similar to above.
\end{itemize}

\paragraph{Case 2.} Suppose $\{b,d\}\nsubseteq \mathcal{C}$. Then either $b\notin \mathcal{C}$ or $d\notin \mathcal{C}$.

\begin{itemize}
\item  If $b \notin \mathcal{C}$, then the rule in Eq.~\eqref{eq:v-rule} implies $v^\ell_{\alpha b} = 0$ whenever $\ell\in\{i,j\}$, for any $\alpha$.  So $\MPEntries=0$ similar to above.

\item  If $d \notin \mathcal{C}$, then the rule in Eq.~\eqref{eq:v-rule} implies $v^\ell_{\alpha d} = 0$ whenever $\ell\in\{i,j\}$, for any $\alpha$.  So $\MPEntries=0$ similar to above.
\end{itemize}
Hence, $\MPEntries=0$ as desired.
\end{proof}

\subsection{Structure of $M''$}

We can now describe the structure of $M''$. The following lemmas characterize the various blocks of the matrix. We will identify many zero entries using the Zero Lemma from above (Lemma~\ref{Lem:GenLemma}). We will also directly compute certain non-zero entries. 

\begin{lemma}[\textbf{Bin 1}] \label{Lem:Bin1_Mixed}
If $(i,j)$ is in Bin 1 and $(a,b,c,d)\in\Lambda$, then the entry $\MPEntries$ is equal to
\[M''_{abcd,ij} = \begin{cases}
z^i_{f(i),g(i)}z^i_{m,k+1}-z^i_{f(i),k+1}z^i_{m,g(i)} & \text{if } \{a,c\}=\{f(i),m\} \text{ and } \{b,d\}=\{g(i),k+1\}, \\
0 & \text{otherwise}.
\end{cases}\]
\end{lemma} 

\begin{remark} \label{rmk:Bin1}
Recall that Bin~1 contains only $(i,j)$ pairs with $i = j > s$. For this reason, only ``non-planted'' $z$ variables appear above. In the other cases below, we will use $v$ variables, which may be either planted or not.
\end{remark}

\begin{proof}
Since $(i,j)$ is in Bin~1, we have $s < i = j \le R$. First consider the case $\{a,c\}=\{f(i),m\}$ and $\{b,d\}=\{g(i),k+1\}$. Since $a<c$ and $b<d$, this forces \[a=f(i),\quad c=m,\quad b=g(i),\quad\text{and}\quad d=k+1.\]
Since $i = j > s$ we have $v^i = z^i$ and $v^j = z^j$, so
\begin{align*}
\MPEntries &= \frac{1}{2}\left(v^i_{ab}v^j_{cd} + v^i_{cd}v^j_{ab} - v^i_{ad}v^j_{cb}-v^i_{cb}v^j_{ad}\right) \\
&= z^i_{ab}z^i_{cd} - z^i_{ad}z^i_{cb} \\
&= z^i_{f(i),g(i)}z^i_{m,k+1} - z^i_{f(i),k+1}z^i_{m,g(i)}.
\end{align*}

Next, let us show that the entries in the other rows are zero. For our $(i,j) = (i,i)$ in Bin 1, the sets $\mathcal{R}$ and $\mathcal{C}$ in the Zero Lemma (Lemma \ref{Lem:GenLemma}) are
\[\mathcal{R} = \{f(i),m\}\] and \[\mathcal{C} = \{g(i),k+1,g(i)+k+1\}.\]
We break into cases depending on the class of $(a,b,c,d)$.
\begin{itemize}
\item Class I: In this class, we have $1\leq b<d\leq k+1$. Since $\{a,c\}\neq\{f(i),m\}$ or $\{b,d\}\neq\{g(i),k+1\}$, we have $\{a,c\}\nsubseteq \mathcal{R}$ or $\{b,d\}\nsubseteq \mathcal{C}$, respectively. Then by Lemma~\ref{Lem:GenLemma}, we have $\MPEntries=0$.

\item Class II: In this class, we have $1\leq b<d\leq k$. 
This means $\{b,d\}\nsubseteq \mathcal{C}$ and thus, $\MPEntries=0$ by Lemma~\ref{Lem:GenLemma}.

\item Class III--V: In these classes, we have $1\leq a<c\leq m-1$.
This means $\{a,c\}\nsubseteq \mathcal{R}$ and thus, $\MPEntries=0$ by Lemma~\ref{Lem:GenLemma}.
\end{itemize}
Hence, the entries $\MPEntries$ in Bin 1 are as claimed.
\end{proof}

\begin{lemma}[\textbf{Bin 2}] \label{Lem:Bin2_Mixed}
If $(i,j)$ is in Bin 2 and $(a,b,c,d)$ is in Class II--V, then the entry $\MPEntries$ is equal to
\[\MPEntries = \frac{1}{2} \begin{cases}
v^i_{f(i),g(i)} v^j_{m,g(j)} - v^i_{m,g(i)} v^j_{f(j),g(j)} & \text{if } \{a,c\}=\{f(i),m\} \text{ and } b=g(i)<g(j)=d, \\
v_{m,g(i)}^i v^j_{f(j),g(j)} - v^i_{f(i),g(i)} v^j_{m,g(j)} & \text{if } \{a,c\}=\{f(i),m\} \text{ and } b=g(j)<g(i)=d, \\
0 & \text{otherwise}.
\end{cases}\]
\end{lemma}

\begin{proof}
Since $(i,j)$ is in Bin~2, we have $i<j$, $f(i)=f(j)$, and $g(i)\neq g(j)$. First, let's look at the non-zero entries of $M''$. 
Choose $a,b,c,d$ such that $\{a,c\}=\{f(i),m\}$, $\{b,d\}=\{g(i),g(j)\}$, $a<c$, and $b<d$. Since $a<c$, this forces
\[a=f(i)\qquad \text{and}\qquad c=m.\]
We have
\[M''_{abcd,ij} = \frac{1}{2} \left( v^i_{f(i),b} v^j_{m,d} + v^i_{m,d} v^j_{f(i),b} - v^i_{f(i),d} v^j_{m,b} - v^i_{m,b} v^j_{f(i),d} \right). \]
We break into cases depending on whether $g(i)$ or $g(j)$ is larger.

\begin{itemize}
\item Say $b=g(i)$ and $d=g(j)$. Then \[\MPEntries = \frac{1}{2} \left( v^i_{f(i),g(i)}v^j_{m,g(j)} + v^i_{m,g(j)}v^j_{f(i),g(i)} -v^i_{f(i),g(j)}v^j_{m,g(i)} - v^i_{m,g(i)}v^j_{f(i),g(j)}. \right) \] By applying the rules~\eqref{eq:v-rule} for when $v^\ell_{\alpha\beta} = 0$ and the assumption $f(i)=f(j)$, we see that
\begin{align*}
M''_{abcd,ij} &= \frac{1}{2} \left( v^i_{f(i),g(i)}v^j_{m,g(j)} + 0\cdot0 - 0\cdot0 - v^i_{m,g(i)}v^j_{f(j),g(j)} \right) \\
&= \frac{1}{2} \left( v^i_{f(i),g(i)}v^j_{m,g(j)} - v^i_{m,g(i)}v^j_{f(j),g(j)} \right).
\end{align*}

\item Say $b=g(j)$ and $d=g(i)$. Then
\[M''_{abcd,ij} = \frac{1}{2} \left( v^i_{f(i),g(j)}v^j_{m,g(i)} + v^i_{m,g(i)}v^j_{f(i),g(j)} - v^i_{f(i),g(i)}v^j_{m,g(j)} - v^i_{m,g(j)}v^j_{f(i),g(i)} \right).\]
By applying the rules~\eqref{eq:v-rule} for when $v^\ell_{\alpha\beta} = 0$ and the assumption $f(i)=f(j)$, we see that 
\begin{align*}
M''_{abcd,ij} &= \frac{1}{2} \left( 0 \cdot 0 + v^i_{m,g(i)}v^j_{f(j),g(j)} - v^i_{f(i),g(i)}v^j_{m,g(j)} - 0 \cdot 0 \right) \\
&= \frac{1}{2} \left( v^i_{m,g(i)}v^j_{f(j),g(j)} - v^i_{f(i),g(i)}v^j_{m,g(j)} \right).
\end{align*}
\end{itemize}

Next, let us show that the entries in the other rows (within Classes II--V) are zero. For our $(i,j)$ in Bin 2, the sets $\mathcal{R}$ and $\mathcal{C}$ in the Zero Lemma (Lemma \ref{Lem:GenLemma}) are \[\mathcal{R} = \{f(i),m\}\] and \[\mathcal{C} = \{g(i),g(j),k+1,g(i)+k+1,g(j)+k+1\}.\]
We break into cases depending on the class of $(a, b, c, d)$.
\begin{itemize}
\item Class II: In this class, we have $1\leq b<d\leq k$. Since $\{a,c\}\neq\{f(i),m\}$ or $\{b,d\}\neq\{g(i),g(j)\}$, we have $\{a,c\}\nsubseteq \mathcal{R}$ or $\{b,d\}\nsubseteq \mathcal{C}$, respectively. Then by Lemma \ref{Lem:GenLemma}, we have $\MPEntries=0$.

\item Class III to V: We have $\{a,c\}=\{f(i),p\}$ for $p\in[m-1]$. Then $\{a,c\}\nsubseteq \mathcal{R}$ and thus, $\MPEntries=0$ by Lemma \ref{Lem:GenLemma}.
\end{itemize}
Hence, the entries $\MPEntries$ in Bin 2 are as claimed.
\end{proof}

\begin{lemma}[\textbf{Bin 3}] \label{Lem:Bin3_Mixed}
If $(i,j)$ is in Bin 3 and $(a,b,c,d)$ is in Class III--V, then the entry $\MPEntries$ is equal to
\[\MPEntries = \frac{1}{2} \begin{cases}
v^i_{f(i),g(i)}v^j_{f(j),k+1} - v^i_{f(i),k+1}v^j_{f(j),g(j)} & \text{if } a=f(i)<f(j)=c \text{ and }  \{b,d\}=\{g(i),k+1\}, \\
v^i_{f(i),k+1}v^j_{f(j),g(j)} - v^i_{f(i),g(i)}v^j_{f(j),k+1} & \text{if } a=f(j)<f(i)=c \text{ and }  \{b,d\}=\{g(i),k+1\}, \\
0 & \text{otherwise}.
\end{cases}\]
\end{lemma}  

\begin{proof}
Since $(i,j)$ is in Bin~3, we have $i<j$, $f(i)\neq f(j)$, and $g(i)=g(j)$. First, let's look at the non-zero entries of $M$.
Choose $a,b,c,d$ such that $\{a,c\}=\{f(i),f(j)\}$, $\{b,d\}=\{g(i),k+1\}$, $a<c$, and $b<d$. Since $b<d$, this forces \[b=g(i)\qquad \text{and}\qquad d=k+1.\] 
We have
\[M''_{abcd,ij} = \frac{1}{2}\left( v^i_{a,g(i)}v^j_{c,k+1} + v^i_{c,k+1}v^j_{a,g(i)} - v^i_{a,k+1}v^j_{c,g(i)} - v^i_{c,g(i)}v^j_{a,k+1} \right).\]
We break into cases depending on whether $f(i)$ or $f(j)$ is larger.
\begin{itemize}
\item Say $a=f(i)$ and $c=f(j)$. Then
\[M''_{abcd,ij} = \frac{1}{2}\left( v^i_{f(i),g(i)}v^j_{f(j),k+1} + v^i_{f(j),k+1}v^j_{f(i),g(i)} -v^i_{f(i),k+1}v^j_{f(j),g(i)} - v^i_{f(j),g(i)}v^j_{f(i),k+1} \right).\] By applying the rules for when $v^\ell_{\alpha\beta} = 0$ and the assumption $g(i)=g(j)$, we see that
\begin{align*}
M''_{abcd,ij} &= \frac{1}{2}\left( v^i_{f(i),g(i)}v^j_{f(j),k+1} + 0\cdot0 - v^i_{f(i),k+1}v^j_{f(j),g(j)} - 0\cdot0 \right) \\
&= \frac{1}{2}\left( v^i_{f(i),g(i)}v^j_{f(j),k+1} - v^i_{f(i),k+1}v^j_{f(j),g(j)} \right).
\end{align*}

\item Say $a=f(j)$ and $c=f(i)$. Then
\[M''_{abcd,ij} = \frac{1}{2}\left( v^i_{f(j),g(i)} v^j_{f(i),k+1} + v^i_{f(i),k+1} v^j_{f(j),g(i)} - v^i_{f(j),k+1} v^j_{f(i),g(i)} - v^i_{f(i),g(i)} v^j_{f(j),k+1} \right) .\] By applying the rules for when $v^\ell_{\alpha\beta} = 0$ and the assumption $g(i)=g(j)$, we see that
\begin{align*}
M''_{abcd,ij} &= \frac{1}{2}\left( 0 \cdot 0 + v^i_{f(i),k+1} \cdot v^i_{f(i),k+1} v^j_{f(j),g(j)} - 0 \cdot 0 - v^i_{f(i),g(i)} v^j_{f(j),k+1} \right) \\
&= \frac{1}{2}\left( v^i_{f(i),k+1} v^j_{f(j),g(j)} - v^i_{f(i),g(i)} v^j_{f(j),k+1} \right).
\end{align*}
\end{itemize}

Next, let us show that the entries in the other rows (within Classes III--V) are zero. For our $(i, j)$ in Bin 3, the sets $\mathcal{R}$ and $\mathcal{C}$ in the Zero Lemma (Lemma \ref{Lem:GenLemma}) are
\[\mathcal{R} = \{f(i),f(j),m\}\] and \[\mathcal{C} = \{g(i),k+1,g(i)+k+1\}.\]
We break into cases depending on the class of $(a, b, c, d)$.
\begin{itemize}
\item Class III: In this class we have $1 \le a < c \le m-1$ and $1 \le b < d = k+1$. Since $\{a,c\}\neq\{f(i),f(j)\}$ or $\{b,d\}\neq\{g(i),k+1\}$, we have $\{a,c\}\nsubseteq \mathcal{R}$ or $\{b,d\}\nsubseteq \mathcal{C}$. Then by Lemma \ref{Lem:GenLemma}, we have $\MPEntries=0$.

\item Class IV to V: We have $\{b,d\}=\{g(i),p\}$ for $p\in[k]$. Then $\{b,d\}\nsubseteq \mathcal{C}$ and thus, $\MPEntries=0$ by Lemma \ref{Lem:GenLemma}.
\end{itemize}
Hence, the entries $\MPEntries$ in Bin 3 are as claimed.
\end{proof}

\begin{lemma}[\textbf{Bin 4}] \label{Lem:Bin4_Mixed}
If $(i,j)$ is in Bin 4 and $(a,b,c,d)$ is in Class IV--V, then the entry $M''_{abcd,ij}$ is equal to
\[\MPEntries = \frac{1}{2} \begin{cases}
v^i_{f(i),g(i)} v^j_{f(j),g(j)} & \text{if } a=f(i),\ c=f(j),\ b=g(i), \text{ and } d=g(j), \\
-v^i_{f(i),g(i)} v^j_{f(j),g(j)} & \text{if } a=f(j),\ c=f(i),\ b=g(i), \text{ and } d=g(j), \\
-v^i_{f(i),g(i)} v^j_{f(j),g(j)} & \text{if } a=f(i),\ c=f(j),\ b=g(j), \text{ and } d=g(i), \\
v^i_{f(i),g(i)} v^j_{f(j),g(j)} & \text{if } a=f(j),\ c=f(i),\ b=g(j), \text{ and } d=g(i), \\ 
v^i_{f(i),g(i)+k+1} v^j_{f(j),g(j)+k+1} & \text{if } a=f(i),\ c=f(j),\ b=g(i)+k+1, \text{ and } d=g(j)+k+1, \\
-v^i_{f(i),g(i)+k+1} v^j_{f(j),g(j)+k+1} & \text{if } a=f(i),\ c=f(j),\ b=g(j)+k+1, \text{ and } d=g(i)+k+1, \\ 
-v^i_{f(i),g(i)+k+1} v^j_{f(j),g(j)+k+1} & \text{if } a=f(j),\ c=f(i),\ b=g(i)+k+1, \text{ and } d=g(j)+k+1, \\ 
v^i_{f(i),g(i)+k+1} v^j_{f(j),g(j)+k+1} & \text{if } a=f(j),\ c=f(i),\ b=g(j)+k+1, \text{ and } d=g(i)+k+1, \\ 0 & \text{otherwise}.
\end{cases}\]
\end{lemma}

\begin{proof} 
Since $(i,j)$ is in Bin~4, we have $i<j$, $f(i)\neq f(j)$, and $g(i)\neq g(j)$. First we consider the non-zero entries, treating Classes IV and V separately. \\

\textbf{Class IV (left side of the matrix).} Choose $a,b,c,d$ such that $\{a,c\}=\{f(i),f(j)\}$, $\{b,d\}=\{g(i),g(j)\}$, $a<c$, and $b<d$. 
Then the entries $\MPEntries$ are as follows.
\begin{itemize}
\item Say $a=f(i)$, $c=f(j)$, $b=g(i)$, and $d=g(j)$. Then
\begin{align*}
M''_{abcd,ij} &= \frac{1}{2}\left( v^i_{f(i),g(i)} v^j_{f(j),g(j)} + v^i_{f(j),g(j)} v^j_{f(i),g(i)} - v^i_{f(i),g(j)} v^j_{f(j),g(i)} - v^i_{f(j),g(i)} v^j_{f(i),g(j)} \right) \\
&= \frac{1}{2}\left( v^i_{f(i),g(i)} v^j_{f(j),g(j)} + 0\cdot0 - 0\cdot0 - 0\cdot0 \right) \\
&= \frac{1}{2}\ v^i_{f(i),g(i)} v^j_{f(j),g(j)}. 
\end{align*} 

\item Say $a=f(i)$, $c=f(j)$, $b=g(j)$, and $d=g(i)$. Then
\begin{align*}
M''_{abcd,ij} &= \frac{1}{2}\left( v^i_{f(i),g(j)} v^j_{f(j),g(i)} + v^i_{f(j),g(i)} v^j_{f(i),g(j)} - v^i_{f(i),g(i)} v^j_{f(j),g(j)} - v^i_{f(j),g(j)} v^j_{f(i),g(i)} \right) \\
&= \frac{1}{2}\left( 0\cdot0 + 0\cdot0 - v^i_{f(i),g(i)} v^j_{f(j),g(j)} - 0\cdot0 \right) \\
&= -\frac{1}{2}\ v^i_{f(i),g(i)} v^j_{f(j),g(j)}. 
\end{align*}

\item Say $a=f(j)$, $c=f(i)$, $b=g(i)$, and $d=g(j)$. Then
\begin{align*}
M''_{abcd,ij} &= \frac{1}{2}\left( v^i_{f(j),g(i)} v^j_{f(i),g(j)} + v^i_{f(i),g(j)} v^j_{f(j),g(i)} - v^i_{f(j),g(j)} v^j_{f(i),g(i)} - v^i_{f(i),g(i)} v^j_{f(j),g(j)} \right) \\
&= \frac{1}{2}\left( 0\cdot0 + 0\cdot0 - 0\cdot0 - v^i_{f(i),g(i)} v^j_{f(j),g(j)} \right) \\
&= -\frac{1}{2}\ v^i_{f(i),g(i)} v^j_{f(j),g(j)}.
\end{align*}

\item Say $a=f(j)$, $c=f(i)$, $b=g(j)$, and $d=g(i)$. Then
\begin{align*}
M''_{abcd,ij} &= \frac{1}{2}\left( v^i_{f(j),g(j)} v^j_{f(i),g(i)} + v^i_{f(i),g(i)} v^j_{f(j),g(j)} - v^i_{f(j),g(i)} v^j_{f(i),g(j)} - v^i_{f(i),g(j)} v^j_{f(j),g(i)} \right) \\
&= \frac{1}{2}\left( 0\cdot0 + v^i_{f(i),g(i)} v^j_{f(j),g(j)} - 0\cdot0 - 0\cdot0 \right) \\
&= \frac{1}{2}\ v^i_{f(i),g(i)} v^j_{f(j),g(j)}.
\end{align*}
\end{itemize}

\textbf{Class V (right side of the matrix).} Choose $a,b,c,d$ such that $\{a,c\}=\{f(i),f(j)\}$, $\{b,d\}=\{g(i)+k+1,g(j)+k+1\}$, $a<c$, and $b<d$. 
Then the entries $\MPEntries$ are as follows.

\begin{itemize}
\item Say $a=f(i)$, $c=f(j)$, $b=g(i)+k+1$, and $d=g(j)+k+1$. Then
\begin{align*}
M''_{abcd,ij} &= \frac{1}{2}\Big( v^i_{f(i),g(i)+k+1} v^j_{f(j),g(j)+k+1} + v^i_{f(j),g(j)+k+1} v^j_{f(i),g(i)+k+1} \\ & \hspace{1cm} - v^i_{f(i),g(j)+k+1} v^j_{f(j),g(i)+k+1} - v^i_{f(j),g(i)+k+1} v^j_{f(i),g(j)+k+1} \Big) \\
&= \frac{1}{2}\left( v^i_{f(i),g(i)+k+1} v^j_{f(j),g(j)+k+1} + 0\cdot0 - 0\cdot0 - 0\cdot0 \right) \\
&= \frac{1}{2}\ v^i_{f(i),g(i)+k+1} v^j_{f(j),g(j)+k+1}. 
\end{align*} 

\item Say $a=f(i)$, $c=f(j)$, $b=g(j)+k+1$, and $d=g(i)+k+1$. Then
\begin{align*}
M''_{abcd,ij} &= \frac{1}{2}\Big( v^i_{f(i),g(j)+k+1} v^j_{f(j),g(i)+k+1} + v^i_{f(j),g(i)+k+1} v^j_{f(i),g(j)+k+1} \\ & \hspace{1cm} - v^i_{f(i),g(i)+k+1} v^j_{f(j),g(j)+k+1} - v^i_{f(j),g(j)+k+1} v^j_{f(i),g(i)+k+1} \Big) \\
&= \frac{1}{2}\left( 0\cdot0 + 0\cdot0 - v^i_{f(i),g(i)+k+1} v^j_{f(j),g(j)+k+1} - 0\cdot0 \right) \\
&= -\frac{1}{2}\ v^i_{f(i),g(i)+k+1} v^j_{f(j),g(j)+k+1}. 
\end{align*}

\item Say $a=f(j)$, $c=f(i)$, $b=g(i)+k+1$, and $d=g(j)+k+1$. Then
\begin{align*}
M''_{abcd,ij} &= \frac{1}{2}\Big( v^i_{f(j),g(i)+k+1} v^j_{f(i),g(j)+k+1} + v^i_{f(i),g(j)+k+1} v^j_{f(j),g(i)+k+1} \\ & \hspace{1cm} - v^i_{f(j),g(j)+k+1} v^j_{f(i),g(i)+k+1} - v^i_{f(i),g(i)+k+1} v^j_{f(j),g(j)+k+1} \Big) \\
&= \frac{1}{2}\left( 0\cdot0 + 0\cdot0 - 0\cdot0 - v^i_{f(i),g(i)+k+1} v^j_{f(j),g(j)+k+1} \right) \\
&= -\frac{1}{2}\ v^i_{f(i),g(i)+k+1} v^j_{f(j),g(j)+k+1}.
\end{align*}

\item Say $a=f(j)$, $c=f(i)$, $b=g(j)+k+1$, and $d=g(i)+k+1$. Then
\begin{align*}
M''_{abcd,ij} &= \frac{1}{2}\Big( v^i_{f(j),g(j)+k+1} v^j_{f(i),g(i)+k+1} + v^i_{f(i),g(i)+k+1} v^j_{f(j),g(j)+k+1} \\ & \hspace{1cm} - v^i_{f(j),g(i)+k+1} v^j_{f(i),g(j)+k+1} - v^i_{f(i),g(j)+k+1} v^j_{f(j),g(i)+k+1} \Big) \\
&= \frac{1}{2}\left(  0\cdot0 + v^i_{f(i),g(i)+k+1} v^j_{f(j),g(j)+k+1} - 0\cdot0 - 0\cdot0 \right) \\
&= \frac{1}{2}\ v^i_{f(i),g(i)+k+1} v^j_{f(j),g(j)+k+1}.
\end{align*}
\end{itemize} 

Next, let us show that the entries in the other rows (within Classes IV--V) are zero. For our $(i, j)$ in Bin 4, the sets $\mathcal{R}$ and $\mathcal{C}$ in the Zero Lemma (Lemma \ref{Lem:GenLemma}) are
\[\mathcal{R} = \{f(i),f(j),m\}\] and \[\mathcal{C} = \{g(i),g(j),k+1,g(i)+k+1,g(j)+k+1\}.\]
We break into cases depending on the class of $(a, b, c, d)$.

\begin{itemize}
\item Class IV: We have $a,c \in [m-1]$ and $b,d \in [k]$. Since $\{a,c\}\neq\{f(i),f(j)\}$ or $\{b,d\}\neq\{g(i),g(i)\}$, we have $\{a,c\}\nsubseteq \mathcal{R}$ or $\{b,d\}\nsubseteq \mathcal{C}$. Thus by Lemma \ref{Lem:GenLemma}, we have $\MPEntries=0$.

\item Class IV: We have $a,c\in[m-1]$ and $b,d\in\{k+2,\ldots,2k+1\}$. Since $\{a,c\}\neq\{f(i),f(j)\}$ or $\{b,d\}\neq\{g(i)+k+1,g(j)+k+1\}$, we have $\{a,c\}\nsubseteq \mathcal{R}$ or $\{b,d\}\nsubseteq \mathcal{C}$. Thus by Lemma \ref{Lem:GenLemma}, we have $\MPEntries=0$.
\end{itemize}

Hence, the entries $\MPEntries$ in Bin 4 are as claimed.
\end{proof}

The resulting structure of $M''$ is illustrated in Figure~\ref{fig:M''Matrix}. Notably, $M''$ is block-upper-triangular with diagonal blocks of size $1 \times 1$ (Bins~1--3) and $2 \times 2$ (Bin~4). To show $\det(M'') \not\equiv 0$, we will need to check that these diagonal blocks have non-zero determinants.

\begin{figure}
    \centering
    \includegraphics[width=1\linewidth]{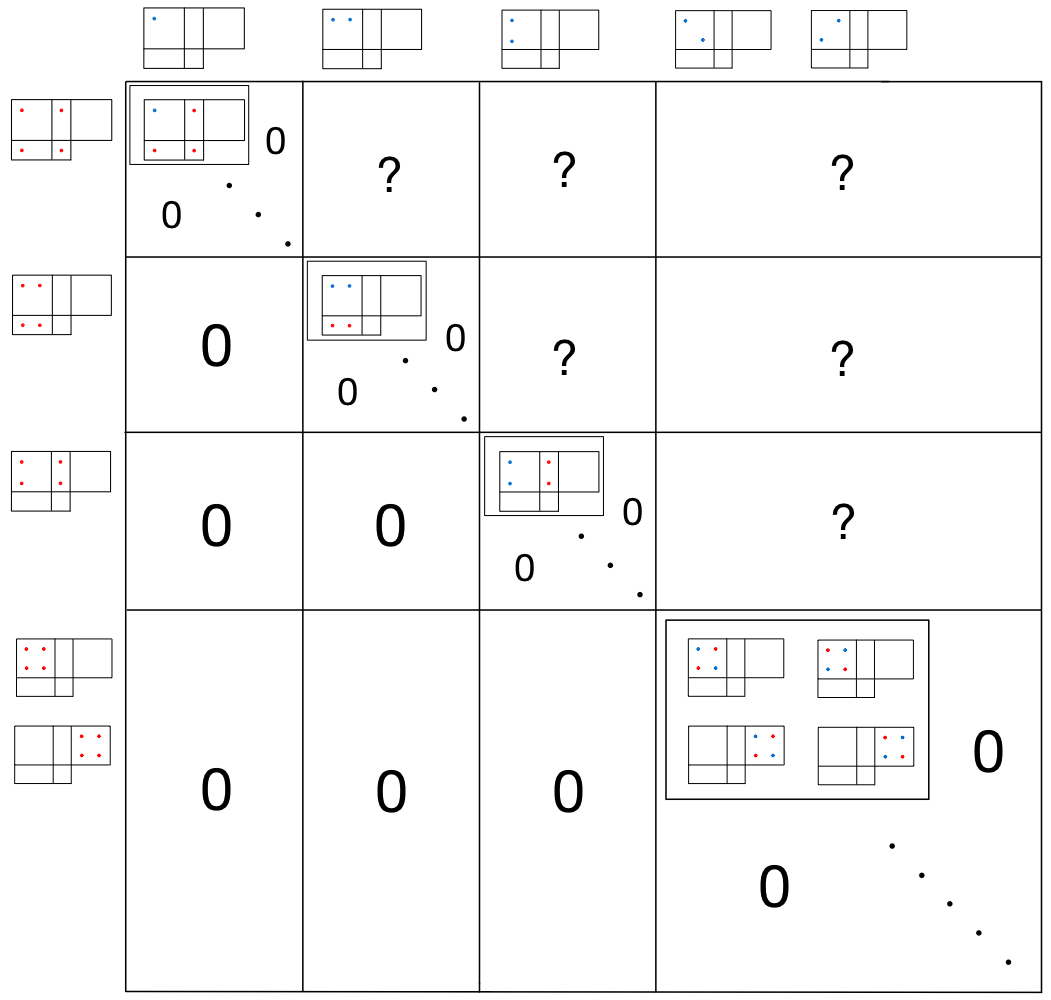}
    \caption{The structure of the matrix $M''$. Columns are indexed by $(i,j)$ pairs, partitioned into Bins 1--4 and visualized as in Figure \ref{fig:Bin1-4}. In Bin~4, each $(i,j)$ is paired with another $(i',j')$. Rows are indexed by tuples $(a,b,c,d)$, partitioned into Classes I--V and visualized as in Figure \ref{fig:ClassI-V}. Classes IV and V are combined in the last block of rows. Each of the first 3 diagonal blocks is a diagonal matrix, with each non-zero entry depicted as a diagram showing how the row and column diagrams overlap. The bottom-right block is a block diagonal matrix with $2 \times 2$ blocks on the diagonal. The blocks with question marks have unknown values that we do not need to control in our proof.
    \label{fig:M''Matrix}
    }
\end{figure}

\subsection{Invertibility of $M''$}

We now show that $M''$ is invertible, building on the previous sections. 

\begin{lemma} \label{Lem:M''DetLemma}
Let $M'' = M''(x,y,z)$ be the matrix from Definition~\ref{Defn:M''}. Then $\det(M'') \not\equiv 0$ as a polynomial in the $x,y,z$ variables.
\end{lemma}

\begin{proof}
Since $M''$ is block-upper-triangular (see Figure~\ref{fig:M''Matrix}), it suffices to show that each diagonal block ($1 \times 1$ or $2 \times 2$) has determinant that is not identically zero.
\begin{itemize}
\item In Block 1 (Bin 1, Class I), by Lemma \ref{Lem:Bin1_Mixed}, the diagonal entries are 
\[z^i_{f(i),g(i)}z^i_{m,k+1} - z^i_{f(i),k+1}z^i_{m,g(i)} \quad\text{for } s<i\leq R.\]
In other words, the row of $M$ with $\{a,c\}=\{f(i),m\}$ and $\{b,d\}=\{g(i),k+1\}$ has the entry above in the $(i,i)$ column. The four variables appearing in this expression are not zeroed out according to the rules in Section \ref{Section:ProofStrategy}. Furthermore, the four variables are distinct, which makes the diagonal entry a non-zero polynomial.

\item In Block 2 (Bin 2, Class II), by Lemma \ref{Lem:Bin2_Mixed}, the diagonal entries are 
\begin{align*}
&\pm \frac{1}{2} \left(v^i_{f(i),g(i)}v^j_{m,g(j)} - v^i_{m,g(i)}v^j_{f(j),g(j)}\right) \\
&\qquad\qquad = \frac{1}{2}
\begin{cases}
\pm\left(z^i_{f(i),g(i)}z^j_{m,g(j)} - z^i_{m,g(i)}z^j_{f(j),g(j)}\right) & \text{if } s< i<j\leq R \\
\pm \left(x^i_{f(i)}y^i_{g(i)} z^j_{m,g(j)} - x^i_{m}y^i_{g(i)} z^j_{f(j),g(j)} \right) & \text{if } 1\leq i\leq s<j\leq R \\
\pm \left(x^i_{f(i)}y^i_{g(i)} x^j_{m}y^j_{g(j)} -\ x^i_{m}y^i_{g(i)} x^j_{f(j)}y^j_{g(j)} \right) & \text{if } 1\leq i<j\leq s.
\end{cases}
\end{align*}
None of the variables appearing here are zeroed out by the rules in Section~\ref{Section:ProofStrategy}.
Furthermore, in each case the two monomials are distinct, which makes the diagonal entry a non-zero polynomial.

\item In Block 3 (Bin 3, Class III), by Lemma~\ref{Lem:Bin3_Mixed}, the diagonal entries are 
\begin{align*}
&\pm \frac{1}{2} \left(v^i_{f(i),g(i)}v^j_{f(j),k+1} - v^i_{f(i),k+1}v^j_{f(j),g(j)}\right) \\
&\qquad\qquad = \frac{1}{2} 
\begin{cases}
\pm \left(z^i_{f(i),g(i)}z^j_{f(j),k+1} - z^i_{f(i),k+1}z^j_{f(j),g(j)}\right) & \text{if } s< i<j\leq R \\
\pm \left(x^i_{f(i)}y^i_{g(i)} z^j_{f(j),k+1} - x^i_{f(i)}y^i_{k+1} z^j_{f(j),g(j)} \right) & \text{if } 1\leq i\leq s <j\leq R \\
\pm \left(x^i_{f(i)}y^i_{g(i)} x^j_{f(j)}y^j_{k+1} - x^i_{f(i)}y^i_{k+1} x^j_{f(j)}y^j_{g(j)} \right) & \text{if } 1\leq i<j\leq s.
\end{cases}
\end{align*}
Similarly to above, these are all non-zero polynomials.

\item In Block 4 (Bin 4, Class IV and V), by Lemma \ref{Lem:Bin3_Mixed}, each $2 \times 2$ diagonal block takes the form
\[A := \frac{1}{2} \begin{bmatrix}
\pm\ v^i_{f(i),g(i)} v^j_{f(j),g(j)} & \pm\ v^{i'}_{f(i'),g(i')} v^{j'}_{f(j'),g(j')} \\ 
\pm\ v^i_{f(i),g(i)+k+1} v^j_{f(j),g(j)+k+1} & \pm\ v^{i'}_{f(i'),g(i')+k+1} v^{j'}_{f(j'),g(j')+k+1}
\end{bmatrix}.
\]
Here, the first column is indexed by $(i,j)$ and the second column is indexed by $(i',j')$ where $\{f(i),f(j)\} = \{f(i'),f(j')\}$ and $\{g(i),g(j)\} = \{g(i'),g(j')\}$. The rows are indexed by the corresponding $(a,b,c,d)$ in Classes IV and V, respectively.
We have

\begin{align} \label{eq:DetAminor}
\det(A) = &\pm\frac{1}{2}\ v^i_{f(i),g(i)} v^j_{f(j),g(j)}v^{i'}_{f(i'),g(i')+k+1} v^{j'}_{f(j'),g(j')+k+1} \nonumber \\
&\hspace{10mm} \pm\frac{1}{2}\ v^{i'}_{f(i'),g(i')} v^{j'}_{f(j'),g(j')} v^i_{f(i),g(i)+k+1} v^j_{f(j),g(j)+k+1}.
\end{align}

To show this is a non-zero polynomial, it suffices to focus on the variables with superscript $i$ and check that these are distinct in each term. We have two cases: $1\leq i\leq s$ and $s< i\leq R$.
\begin{enumerate}
\item If $1 \leq i \leq s$, then \[\pm\ v^i_{f(i),g(i)} = \pm\ x^i_{f(i)}y^i_{g(i)} \quad\text{and}\quad \pm v^i_{f(i),g(i)+k+1} = \pm\ x^i_{f(i)}y^i_{g(i)+k+1}\] in the first term and second term in \eqref{eq:DetAminor}, respectively. Since the underlying $y$ variables are distinct in these terms, we have $\det(A)\not\equiv0$.

\item If $s< i \leq R$, then \[\pm\ v^i_{f(i),g(i)} = \pm\ z^i_{f(i),g(i)} \quad\text{and}\quad \pm v^i_{f(i),g(i)+k+1} = \pm\ z^i_{f(i),g(i)+k+1}\] in the first term and second term in \eqref{eq:DetAminor}, respectively. Since the underlying $z$ variables are distinct in these terms, we have $\det(A)\not\equiv0$.
\end{enumerate}
So, the determinant of each $2 \times 2$ diagonal block within Block 4 is non-zero as a polynomial in $x,y,z$.
\end{itemize}
Therefore, $\det(M'') \not\equiv 0$ as desired.
\end{proof}

\section{Proofs of Main Results}

\subsection{Positive Result}
\label{subsec:PosResult}

We now prove our main positive result.

\begin{proof}[Proof of Theorem \ref{Thm:GoalHolds}]
Let $M$ be the matrix defined in Section \ref{Section:Prelims&Perspective} and let $M''$ be the square submatrix constructed in Section \ref{Section:Construction}. Recall from Section~\ref{Section:ProofStrategy} that the construction of $M''$ requires $R \le \frac{1}{2}(m-1)(n-2)$. By Lemma \ref{Lem:M''DetLemma}, $\det(M'') \not\equiv 0$ as a polynomial in $x,y,z$. This means $M$ has full column rank for generically chosen $x,y,z$. Hence, \eqref{eq:goal} holds by Lemma \ref{Lem:GoalLemma} as desired.
\end{proof}

\subsection{Negative Result}
\label{subsec:NegResult}

We now prove our main negative result.

\begin{proof}[Proof of Theorem \ref{Thm:JLVFail}]
Suppose \eqref{eq:goal} holds. We wish to show that $\binom{m}{2}\binom{n}{2}\geq\binom{R+1}{2} - s$.
Observe that
\begin{align*}
\dim(S^2(\R^{mn})) &\geq \dim(S^2(\U) \cup \SpEBasis) \\
&=\dim(S^2(\U)) + \dim(\SpEBasis) - \dim(S^2(\U)\cap \SpEBasis),
\end{align*}
which yields
\begin{align*}
\binom{mn+1}{2} &\geq \binom{R+1}{2} + \left[\binom{mn+1}{2} - \binom{m}{2}\binom{n}{2}\right] - \dim(S^2(\U)\cap \SpEBasis).
\end{align*}
So,
\begin{equation}\label{eq:intersect_lowerbound}
\dim(S^2(\U)\cap \SpEBasis) \geq \binom{R+1}{2} - \binom{m}{2}\binom{n}{2}.
\end{equation}
By our assumption of \eqref{eq:goal}, the left-hand side above is equal to $\dim(\Span\{(v^1)^{\otimes 2},\ldots,(v^s)^{\otimes 2}\})$, which is at most $s$ (in fact, exactly $s$, since $v^1,\ldots,v^s$ are linearly independent by Lemma~\ref{lem:lin-indep}). Hence, we have our desired inequality.
\end{proof}

\section{Conjecture and Supporting Results}
\label{Section:Conj&Evidence}

\subsection{Conjecure on the Exact Threshold}

In Theorem \ref{Thm:JLVFail}, we gave a sufficient condition for the JLV algorithm to fail. Numerical tests suggest that the JLV algorithm works right up until this breaking point, leading to the following conjecture.
\begin{conjecture}
\label{Conjecture:SharpBoundary}
If $\text{ }\U = \Span\{v^1,\ldots,v^R\} \subseteq \R^{m\times n}$ is a generically chosen planted subspace with parameters $0 \le s \le R$ where
\begin{equation}\label{eq:conj-cond}
\binom{m}{2}\binom{n}{2} \ge \binom{R+1}{2} - s,
\end{equation}
then \eqref{eq:goal} holds, and as a consequence the JLV algorithm succeeds.
\end{conjecture}

The left-hand side of~\eqref{eq:conj-cond} counts the number of rows of the matrix $M$ from Section~\ref{Section:Prelims&Perspective}, and the right-hand side counts the number of columns of $M$. Therefore,~\eqref{eq:conj-cond} is the condition under which $M$ could plausibly have full column rank, which is equivalent to \eqref{eq:goal}.

We remark that this conjecture is known to hold in the special case $s=0$ and $\min\{m,n\}=2$ \cite[Theorem~1.10]{DJL24}. In fact, that result shows that for \emph{any} subspace that avoids the rank-1 matrices (not just a generic one), the Nullstellensatz hierarchy can certify this at level $\min\{m,n\}$ (which has polynomial time if $\min\{m,n\}$ is constant).
We will provide computer-aided proof for a finite set of new cases, by very different techniques.

In Section \ref{Section:Setting+MainResults}, we introduced the concept of \emph{genericity} and demonstrated via Lemma \ref{Lem:GoalLemma} that $\eqref{eq:goal}$ will be a generic property of subspaces of a given type if it holds true.
As such, we need only produce a single subspace for which the matrix $M$ defined in Section~\ref{Section:Prelims&Perspective} has full column rank.
We will term such subspaces produced, \emph{witness} subspaces or just \emph{witnesses} since they ``witness'' the truth of the conjecture.
One can certify a subspace is a witness by computing some maximal minor (i.e., a subdeterminant involving all columns) of the $M$ matrix and verifying it is not zero.
We call the value of such a non-zero maximal minor (and the rows of $M$ to which it corresponds) a \emph{certificate} for the proof. In exact arithmetic, producing such a single witness-certificate pair constitutes a computer-aided proof that $\eqref{eq:goal}$ holds generically for a specific choice of parameters $m,n,s,R$ (since this implies that some maximal minor is non-zero as a polynomial).
Using this observation we have proven Conjecture~\ref{Conjecture:SharpBoundary} in 3703 specific cases.

\begin{theorem}
    \label{thm:certifiedAll}
    For all choices of $m,n \ge 2$ such that $\frac1{\sqrt{2}}mn\leq 60$, Conjecture~\ref{Conjecture:SharpBoundary} holds for every $s$ and $R$ satisfying the given inequality~\eqref{eq:conj-cond}.
\end{theorem}

\begin{theorem}
    \label{thm:certifiedParticular}
    For all choices of $m,n \ge 2$ such that $\frac1{\sqrt{2}}mn\leq 120$, Conjecture~\ref{Conjecture:SharpBoundary} holds whenever $s \in \{0,R\}$ and the given inequality~\eqref{eq:conj-cond} holds.
\end{theorem}

The values of $m$ and $n$ for which these theorems apply and additional relevant details appear in Table~\ref{table:plant_max_values}.
We produce witnesses which are particular subspaces spanned by integer-valued basis vectors and certificates as described above, where the determinants are computed in modular arithmetic to verify non-singularity.
Further details are given in Section~\ref{subsec:ComputerAidedProofs}.
Witness-certificate pairs are recorded in our supplemental material described in Section~\ref{sec:Code&Data}.

In addition, we have implemented the JLV algorithm in floating-point arithmetic to verify that it can recover planted rank-1 matrices up to a small error attributable to floating-point precision loss. 
These tests run quicker than our method for producing rigorous certificates, allowing us to run larger problem sizes and give strong evidence for an additional 3289 cases of Conjecture~\ref{Conjecture:SharpBoundary}. The supporting code and data are again discussed in Section~\ref{sec:Code&Data}. These data lead to the theorems below marked with an asterisk (*).
Since computations were done in floating-point arithmetic with possible loss of precision, we have not rigorously proven that the $M$ matrix is full-rank; however, we believe the experiments justify the truth of these statements beyond a reasonable doubt.

\begin{theorem-star}
    \label{thm:numericalAll}
    For all choices of $m,n \ge 2$ such that $\frac1{\sqrt{2}}mn\leq 80$, Conjecture~\ref{Conjecture:SharpBoundary} holds for every $s$ and $R$ satisfying the given inequality~\eqref{eq:conj-cond}.
\end{theorem-star}

\begin{theorem-star}
    \label{thm:numericalParticular}
    For all choices of $m,n \ge 2$ such that $\frac1{\sqrt{2}}mn\leq 165$, Conjecture~\ref{Conjecture:SharpBoundary} holds whenever $s \in \{0,R\}$ and the given inequality~\eqref{eq:conj-cond} holds.
\end{theorem-star}

These results are all summarized in Table~\ref{table:plant_max_values}, with a more in-depth discussion deferred to Section~\ref{subsec:NumericalExperiments}. We also confirm that for $R$ and $s$ beyond the inequality given in Conjecture~\ref{Conjecture:SharpBoundary}, the JLV algorithm fails to recover any planted rank-1 matrices (or fails to certify there are no such matrices in the case $s=0$). This too is discussed in Section~\ref{subsec:NumericalExperiments}.

\begin{table} 
\label{table:proofAndNumericalResults}
\centering
\label{table:plant_max_values}
\begin{tabular}{c|*{14}{c}}
$m \backslash n$ & 2 & 3 & 4 & 5 & 6 & 7 & 8 & 9 & 10 & 11 & 12 & 13 & 14 & 15 \\ \hline 
2 & \cellcolor{myblue}2 & \color{lightgray}3 & \color{lightgray}4 & \color{lightgray}5 & \color{lightgray}6 & \color{lightgray}7 & \color{lightgray}8 & \color{lightgray}9 & \color{lightgray}10 & \color{lightgray}11 & \color{lightgray}12 & \color{lightgray}13 & \color{lightgray}14 & \color{lightgray}15 \\
3 & \cellcolor{myblue}3 & \cellcolor{myblue}4 & \color{lightgray}6 & \color{lightgray}8 & \color{lightgray}10 & \color{lightgray}11 & \color{lightgray}13 & \color{lightgray}15 & \color{lightgray}16 & \color{lightgray}18 & \color{lightgray}20 & \color{lightgray}22 & \color{lightgray}23 & \color{lightgray}25 \\
4 & \cellcolor{myblue}4 & \cellcolor{myblue}6 & \cellcolor{myblue}9 & \color{lightgray}11 & \color{lightgray}13 & \color{lightgray}16 & \color{lightgray}18 & \color{lightgray}21 & \color{lightgray}23 & \color{lightgray}26 & \color{lightgray}28 & \color{lightgray}31 & \color{lightgray}33 & \color{lightgray}36 \\
5 & \cellcolor{myblue}5 & \cellcolor{myblue}8 & \cellcolor{myblue}11 & \cellcolor{myblue}14 & \color{lightgray}17 & \color{lightgray}21 & \color{lightgray}24 & \color{lightgray}27 & \color{lightgray}30 & \color{lightgray}33 & \color{lightgray}36 & \color{lightgray}40 & \color{lightgray}43 & \color{lightgray}46 \\
6 & \cellcolor{myblue}6 & \cellcolor{myblue}10 & \cellcolor{myblue}13 & \cellcolor{myblue}17 & \cellcolor{myblue}21 & \color{lightgray}25 & \color{lightgray}29 & \color{lightgray}33 & \color{lightgray}37 & \color{lightgray}41 & \color{lightgray}45 & \color{lightgray}48 & \color{lightgray}52 & \color{lightgray}56 \\
7 & \cellcolor{myblue}7 & \cellcolor{myblue}11 & \cellcolor{myblue}16 & \cellcolor{myblue}21 & \cellcolor{myblue}25 & \cellcolor{myblue}30 & \color{lightgray}34 & \color{lightgray}39 & \color{lightgray}43 & \color{lightgray}48 & \color{lightgray}53 & \color{lightgray}57 & \color{lightgray}62 & \color{lightgray}66 \\
8 & \cellcolor{myblue}8 & \cellcolor{myblue}13 & \cellcolor{myblue}18 & \cellcolor{myblue}24 & \cellcolor{myblue}29 & \cellcolor{myblue}34 & \cellcolor{myblue}40 & \color{lightgray}45 & \color{lightgray}50 & \color{lightgray}56 & \color{lightgray}61 & \color{lightgray}66 & \color{lightgray}71 & \color{lightgray}77 \\
9 & \cellcolor{myblue}9 & \cellcolor{myblue}15 & \cellcolor{myblue}21 & \cellcolor{myblue}27 & \cellcolor{myblue}33 & \cellcolor{myblue}39 & \cellcolor{myblue}45 & \cellcolor{myblue}51 & \color{lightgray}57 & \color{lightgray}63 & \color{lightgray}69 & \color{lightgray}75 & \color{lightgray}81 & \color{lightgray}87 \\
10 & \cellcolor{myblue}10 & \cellcolor{myblue}16 & \cellcolor{myblue}23 & \cellcolor{myblue}30 & \cellcolor{myblue}37 & \cellcolor{myblue}43 & \cellcolor{myblue}50 & \cellcolor{mypurple}57 & \cellcolor{mypurple}64 & \color{lightgray}70 & \color{lightgray}77 & \color{lightgray}84 & \color{lightgray}91 & \color{lightgray}97 \\
11 & \cellcolor{myblue}11 & \cellcolor{myblue}18 & \cellcolor{myblue}26 & \cellcolor{myblue}33 & \cellcolor{myblue}41 & \cellcolor{myblue}48 & \cellcolor{mypurple}56 & \cellcolor{mypurple}63 & \cellcolor{mypurple}70 & \cellcolor{mymagenta}78 & \color{lightgray}85 & \color{lightgray}93 & \color{lightgray}100 & \color{lightgray}107 \\
12 & \cellcolor{myblue}12 & \cellcolor{myblue}20 & \cellcolor{myblue}28 & \cellcolor{myblue}36 & \cellcolor{myblue}45 & \cellcolor{myblue}53 & \cellcolor{mypurple}61 & \cellcolor{mypurple}69 & \cellcolor{mymagenta}77 & \cellcolor{mymagenta}85 & \cellcolor{mymagenta}93 & \color{lightgray}101 & \color{lightgray}110 & \color{lightgray}118 \\
13 & \cellcolor{myblue}13 & \cellcolor{myblue}22 & \cellcolor{myblue}31 & \cellcolor{myblue}40 & \cellcolor{myblue}48 & \cellcolor{mypurple}57 & \cellcolor{mypurple}66 & \cellcolor{mymagenta}75 & \cellcolor{mymagenta}84 & \cellcolor{mymagenta}93 & \cellcolor{mymagenta}101 & \cellcolor{mymagenta}110 & \color{lightgray}119 & \color{lightgray}128 \\
14 & \cellcolor{myblue}14 & \cellcolor{myblue}23 & \cellcolor{myblue}33 & \cellcolor{myblue}43 & \cellcolor{myblue}52 & \cellcolor{mypurple}62 & \cellcolor{mypurple}71 & \cellcolor{mymagenta}81 & \cellcolor{mymagenta}91 & \cellcolor{mymagenta}100 & \cellcolor{mymagenta}110 & \cellcolor{myorange}119 & \cellcolor{myorange}129 & \color{lightgray}138 \\
15 & \cellcolor{myblue}15 & \cellcolor{myblue}25 & \cellcolor{myblue}36 & \cellcolor{myblue}46 & \cellcolor{mypurple}56 & \cellcolor{mypurple}66 & \cellcolor{mymagenta}77 & \cellcolor{mymagenta}87 & \cellcolor{mymagenta}97 & \cellcolor{mymagenta}107 & \cellcolor{myorange}118 & \cellcolor{myorange}128 & \cellcolor{myorange}138 & \cellcolor{myorange}148 \\
16 & \cellcolor{myblue}16 & \cellcolor{myblue}27 & \cellcolor{myblue}38 & \cellcolor{myblue}49 & \cellcolor{mypurple}60 & \cellcolor{mypurple}71 & \cellcolor{mymagenta}82 & \cellcolor{mymagenta}93 & \cellcolor{mymagenta}104 & \cellcolor{myorange}115 & \cellcolor{myorange}126 & \cellcolor{myorange}137 & \cellcolor{myorange}148 & 159 \\
17 & \cellcolor{myblue}17 & \cellcolor{myblue}29 & \cellcolor{myblue}40 & \cellcolor{mypurple}52 & \cellcolor{mypurple}64 & \cellcolor{mymagenta}76 & \cellcolor{mymagenta}87 & \cellcolor{mymagenta}99 & \cellcolor{myorange}111 & \cellcolor{myorange}122 & \cellcolor{myorange}134 & \cellcolor{myorange}146 & 157 & 169 \\
18 & \cellcolor{myblue}18 & \cellcolor{myblue}30 & \cellcolor{myblue}43 & \cellcolor{mypurple}55 & \cellcolor{mypurple}68 & \cellcolor{mymagenta}80 & \cellcolor{mymagenta}93 & \cellcolor{mymagenta}105 & \cellcolor{myorange}117 & \cellcolor{myorange}130 & \cellcolor{myorange}142 & 154 & 167 & 179 \\
19 & \cellcolor{myblue}19 & \cellcolor{myblue}32 & \cellcolor{myblue}45 & \cellcolor{mypurple}58 & \cellcolor{mymagenta}72 & \cellcolor{mymagenta}85 & \cellcolor{mymagenta}98 & \cellcolor{myorange}111 & \cellcolor{myorange}124 & \cellcolor{myorange}137 & \cellcolor{myorange}150 & 163 & 176 & 190 \\
20 & \cellcolor{myblue}20 & \cellcolor{myblue}34 & \cellcolor{myblue}48 & \cellcolor{mypurple}62 & \cellcolor{mymagenta}76 & \cellcolor{mymagenta}89 & \cellcolor{mymagenta}103 & \cellcolor{myorange}117 & \cellcolor{myorange}131 & \cellcolor{myorange}145 & 158 & 172 & 186 & 200 \\
21 & \cellcolor{myblue}21 & \cellcolor{myblue}36 & \cellcolor{myblue}50 & \cellcolor{mypurple}65 & \cellcolor{mymagenta}79 & \cellcolor{mymagenta}94 & \cellcolor{mymagenta}108 & \cellcolor{myorange}123 & \cellcolor{myorange}137 & \cellcolor{myorange}152 & 166 & 181 & 196 & 210 \\
22 & \cellcolor{myblue}22 & \cellcolor{myblue}37 & \cellcolor{mypurple}53 & \cellcolor{mypurple}68 & \cellcolor{mymagenta}83 & \cellcolor{mymagenta}99 & \cellcolor{myorange}114 & \cellcolor{myorange}129 & \cellcolor{myorange}144 & 159 & 175 & 190 & 205 & 220 \\
23 & \cellcolor{myblue}23 & \cellcolor{myblue}39 & \cellcolor{mypurple}55 & \cellcolor{mymagenta}71 & \cellcolor{mymagenta}87 & \cellcolor{mymagenta}103 & \cellcolor{myorange}119 & \cellcolor{myorange}135 & \cellcolor{myorange}151 & 167 & 183 & 199 & 215 & 231 \\
24 & \cellcolor{myblue}24 & \cellcolor{myblue}41 & \cellcolor{mypurple}58 & \cellcolor{mymagenta}74 & \cellcolor{mymagenta}91 & \cellcolor{mymagenta}108 & \cellcolor{myorange}124 & \cellcolor{myorange}141 & 158 & 174 & 191 & 208 & 224 & 241 \\
25 & \cellcolor{myblue}25 & \cellcolor{myblue}42 & \cellcolor{mypurple}60 & \cellcolor{mymagenta}77 & \cellcolor{mymagenta}95 & \cellcolor{myorange}112 & \cellcolor{myorange}130 & \cellcolor{myorange}147 & 164 & 182 & 199 & 216 & 234 & 251 \\
26 & \cellcolor{myblue}26 & \cellcolor{myblue}44 & \cellcolor{mypurple}62 & \cellcolor{mymagenta}81 & \cellcolor{mymagenta}99 & \cellcolor{myorange}117 & \cellcolor{myorange}135 & 153 & 171 & 189 & 207 & 225 & 243 & 261 \\
27 & \cellcolor{myblue}27 & \cellcolor{myblue}46 & \cellcolor{mypurple}65 & \cellcolor{mymagenta}84 & \cellcolor{mymagenta}103 & \cellcolor{myorange}121 & \cellcolor{myorange}140 & 159 & 178 & 196 & 215 & 234 & 253 & 271 \\
28 & \cellcolor{myblue}28 & \cellcolor{myblue}48 & \cellcolor{mypurple}67 & \cellcolor{mymagenta}87 & \cellcolor{mymagenta}106 & \cellcolor{myorange}126 & \cellcolor{myorange}145 & 165 & 184 & 204 & 223 & 243 & 262 & 282 \\
29 & \cellcolor{myblue}29 & \cellcolor{mypurple}49 & \cellcolor{mymagenta}70 & \cellcolor{mymagenta}90 & \cellcolor{myorange}110 & \cellcolor{myorange}131 & \cellcolor{myorange}151 & 171 & 191 & 211 & 232 & 252 & 272 & 292 \\
30 & \cellcolor{myblue}30 & \cellcolor{mypurple}51 & \cellcolor{mymagenta}72 & \cellcolor{mymagenta}93 & \cellcolor{myorange}114 & \cellcolor{myorange}135 & 156 & 177 & 198 & 219 & 240 & 261 & 281 & 302 \\
31 & \cellcolor{myblue}31 & \cellcolor{mypurple}53 & \cellcolor{mymagenta}75 & \cellcolor{mymagenta}96 & \cellcolor{myorange}118 & \cellcolor{myorange}140 & 161 & 183 & 205 & 226 & 248 & 269 & 291 & 312 \\
32 & \cellcolor{myblue}32 & \cellcolor{mypurple}55 & \cellcolor{mymagenta}77 & \cellcolor{mymagenta}100 & \cellcolor{myorange}122 & \cellcolor{myorange}144 & 167 & 189 & 211 & 234 & 256 & 278 & 300 & 323 \\
33 & \cellcolor{myblue}33 & \cellcolor{mypurple}56 & \cellcolor{mymagenta}80 & \cellcolor{mymagenta}103 & \cellcolor{myorange}126 & \cellcolor{myorange}149 & 172 & 195 & 218 & 241 & 264 & 287 & 310 & 333 \\
34 & \cellcolor{myblue}34 & \cellcolor{mypurple}58 & \cellcolor{mymagenta}82 & \cellcolor{myorange}106 & \cellcolor{myorange}130 & 154 & 177 & 201 & 225 & 248 & 272 & 296 & 320 & 343 \\
35 & \cellcolor{myblue}35 & \cellcolor{mypurple}60 & \cellcolor{mymagenta}85 & \cellcolor{myorange}109 & \cellcolor{myorange}134 & 158 & 183 & 207 & 231 & 256 & 280 & 305 & 329 & 353 \\
36 & \cellcolor{myblue}36 & \cellcolor{mypurple}61 & \cellcolor{mymagenta}87 & \cellcolor{myorange}112 & \cellcolor{myorange}137 & 163 & 188 & 213 & 238 & 263 & 288 & 313 & 339 & 364 \\
37 & \cellcolor{myblue}37 & \cellcolor{mypurple}63 & \cellcolor{mymagenta}89 & \cellcolor{myorange}115 & \cellcolor{myorange}141 & 167 & 193 & 219 & 245 & 271 & 297 & 322 & 348 & 374 \\
38 & \cellcolor{myblue}38 & \cellcolor{mymagenta}65 & \cellcolor{mymagenta}92 & \cellcolor{myorange}119 & \cellcolor{myorange}145 & 172 & 198 & 225 & 252 & 278 & 305 & 331 & 358 & 384 \\
39 & \cellcolor{myblue}39 & \cellcolor{mymagenta}67 & \cellcolor{mymagenta}94 & \cellcolor{myorange}122 & 149 & 176 & 204 & 231 & 258 & 286 & 313 & 340 & 367 & 394 \\
40 & \cellcolor{myblue}40 & \cellcolor{mymagenta}68 & \cellcolor{mymagenta}97 & \cellcolor{myorange}125 & 153 & 181 & 209 & 237 & 265 & 293 & 321 & 349 & 377 & 405
\end{tabular}
\caption{The value of $R_\text{max}(m,n)$ for selected $m$ and $n$. 
Grayed entries are redundant. Cases for which we have proved Conjecture~\ref{Conjecture:SharpBoundary} for all $0\leq s\leq R$ (Thm~\ref{thm:certifiedAll}) are marked in \colorbox{myblue}{blue}, and cases where we extended via strong numerical evidence (Thm*~\ref{thm:numericalAll}) are marked in \colorbox{mypurple}{purple}.
Cases where we have proof in the special cases $s=0$ and $s=R$ (Thm~\ref{thm:certifiedParticular}) are marked in \colorbox{mymagenta}{magenta} (also including the blue and purple regions), and these likewise extend to \colorbox{myorange}{orange} via floating-point calculation (Thm*~\ref{thm:numericalParticular}).
Not all cases in columns $n=2$ to $5$ are listed due to limited space.} 
\end{table}

\subsection{Determining Boundary Cases}
To describe the evidence we have for Conjecture~\ref{Conjecture:SharpBoundary}, we first give methods for determining the maximum subspace dimension $R$ given $m,n \ge 2$ and $s \ge 0$ for which the condition of the conjecture will hold.
\begin{definition}
\label{def:Rmax}
The maximum $R$ for which~\eqref{eq:conj-cond} holds is
\[
R_\mathrm{max}(s) = R_\mathrm{max}(s,m,n) := \left\lfloor \sqrt{\frac12 m(m-1)n(n-1) + \frac14 + 2s} - \frac12 \right\rfloor.
\]
\end{definition}
However, we must have $s\leq R$ which restricts the possible values for $s$.
\begin{definition}
To this end, define
\[
R_\mathrm{max} = R_\mathrm{max}(m,n) := \max\: \{ s \in \Z^+ : s\leq R_\mathrm{max}(s,m,n) \},
\]
which is the maximum dimension $R$ for the case $s=R$.
\end{definition}
The possible values for $s$ are $\{0,1,\dots, R_\text{max}\}$.
We can explicitly solve to find
\[
    R_\mathrm{max} = \left\lfloor \sqrt{\frac12 m(m-1)n(n-1) + \frac14} + \frac12 \right\rfloor .
\]

\begin{remark}
As a function, $R_\mathrm{max}(s)$ is monotone increasing in $s$ and has a maximum and minimum that differ by exactly one.
Hence, it has the form
\[
R_\mathrm{max}(s) = \begin{cases}
    R_\mathrm{max} & s^* \le s \le R_\mathrm{max} \\
    R_\mathrm{max} - 1& 0 \le s < s^* \\
\end{cases}, \quad s\in\{0,1,2,\dots, R_\mathrm{max}\}
\]
where $s^* \in \{0,1,\ldots,R_\mathrm{max}\}$ is some threshold depending on $m$ and $n$.
\end{remark}

In summary, for given $m$ and $n$, the $(s,R)$ pairs satisfying~\eqref{eq:conj-cond} are those with $s \in \{0,\ldots,R_\mathrm{max}\}$ and $R \in \{s,\ldots,R_\mathrm{max}(s)\}$.

\subsection{Computer-Aided Proofs}
\label{subsec:ComputerAidedProofs}

Computing a determinant in exact arithmetic may require impractical amounts of memory. Therefore, to construct our certificates, we will produce a subspace spanned by integer-valued basis vectors such that some maximal minor of the matrix $M$  has non-zero determinant as a polynomial over the finite field $\F_p \cong \Z/p\Z$ for some prime number $p$. If the determinant is non-zero over $\F_p$ it is also non-zero over $\R$, hence our exact computations constitute proof of the conjecture for specific problem sizes.

\subsubsection{Methods}
Here we detail the constructions and computations we do to algorithmically prove Conjecture~\ref{Conjecture:SharpBoundary} for a finite set of cases. Notice that for planted matrices of the form $v^i=x^i\otimes y^i$, $i=1,\dots,s$, a necessary condition for identifiability of the planted matrices is that each set of factors $\{x^i\}_{i=1}^s \subset \R^m$ and $\{y_i\}_{i=1}^s \subset \R^n$ must be pairwise linearly independent (i.e., no two vectors are scalar multiples).
If, for example, $x^i$ is a scalar multiple of $x^j$ then $\{x^i\otimes y^i + \alpha (x^j\otimes y^j) : \alpha \in \R\}$ will be an infinite family of rank-1 matrices in the subspace $\U$.
A set of vectors drawn from an absolutely continuous distribution on $\R^d$ will be pairwise linearly independent with probability one, but drawing vectors from $\F_p^d$ will potentially violate this with high probability, depending on $p,d$ and the method of sampling. Over $\F_p^d$, there are exactly $(p^d-1)/(p-1)$ pairwise linearly independent vectors. We chose $p=997$, which is large enough to ensure that pairwise linearly independent sets exist for all values of $m,n,s$ that we consider.

To produce a set of pairwise linearly independent vectors over $\F_p^d$, we sample a vector uniformly at random from $\F_p^d$, but only add it to the set if it is not the zero vector or a multiple modulo $p$ of a vector in the current set. The process is repeated until the set contains $s$ elements. It follows that $\{x^i\}_{i=1}^s$ is uniformly distributed on $(\F^m_p)^s$ conditional on being pairwise linearly independent. The same procedure is used to generate $\{y^i\}_{i=1}^s$. The remaining matrices, $\{v^i\}_{i=s+1}^R$, we sample uniformly at random from $\F^{m\times n}_p$. We will see that this procedure successfully produces the witnesses we need, with high probability.

To determine whether the subspace constitutes a witness to the truth of Conjecture~\ref{Conjecture:SharpBoundary}, we construct the matrix $M$ given at the beginning of Section~\ref{Section:Prelims&Perspective} with computations done over $\F_p$.
To find a non-zero maximal minor if one exists over $\F_p$, we run Gauss--Jordan elimination in modular arithmetic on $M$ with row pivoting.
By recording the locations of rows after elimination and pivoting, we may report which rows of the $M$ matrix correspond to the maximal minor.
The product of terms on the diagonal modulo $p$ after elimination gives the value of the minor modulo $p$ and either certifies that the subspace $\U=\Span\{v_1,v_2,\dots,v_R\}\subset \R^{m \times n}$ produces an $M$ matrix of full rank if the minor is non-zero, or fails to certify that $M$ has full rank. If certification fails, we resample all matrix factors and matrices and proceed with a new subspace.

\subsubsection{Results}
We were able to produce pseudo-random witness subspaces on the first initialization (pseudo-random seed set to 0) in all 3703 cases using a modulus of $p=997$.
See Section~\ref{sec:Code&Data} for a description of this data.
These certificates constitute proof of Theorem~\ref{thm:certifiedAll} and Theorem~\ref{thm:certifiedParticular}.
We note that it appears to be possible to generate witness subspaces for any $p$ provided there is a pairwise linearly independent set i.e., $(p^{\min\{m,n\}}-1)/(p-1) \ge s$, but such subspaces are more abundant for larger $p$.

\subsection{Numerical Experiments}
\label{subsec:NumericalExperiments}

Over modular arithmetic we can certify that a subspace $\U$ produces a non-singular $M$ matrix, implying success of the JLV algorithm on generic inputs; however, we also aim to confirm numerically that the full JLV algorithm works correctly over $\R$ in recovering planted rank-1 matrices. To verify this, we ran a number of computations in floating-point arithmetic and report the results here.
The routines employed to do linear algebra in floating-point arithmetic run much quicker than Gauss--Jordan elimination, so we verified numerically that the JLV algorithm succeeds even in situations where we have not produced a computer-aided proof of its correctness.
This gives strong evidence for cases of the conjecture indicated in Table~\ref{table:proofAndNumericalResults} in purple and orange.

In addition to these checks, to understand the JLV algorithm's failure modes, we ran the algorithm on random subspaces \emph{beyond} where \eqref{eq:goal} holds but where planted solutions are still identifiable (see Section~\ref{Section:Introduction}). For this, we tested all matrix shapes $(m, n)$ with $2 \leq m \leq n$ and $mn/\sqrt{2} < 60$.
For each matrix shape $(m,n)$ we tested cases with parameters $s,R$ where $0\leq s \leq R_\mathrm{max}(m,n)+1$ and $R = R_\mathrm{max}(m,n,s) + 1$ as long as $R \leq (m-1)(n-1)$, i.e.\ one over the conjectured bound but still generically identifiable.
These cases represent all choices of $s,R$ that are the smallest that exceed the inequality~\eqref{eq:conj-cond} for given $m$ and $n$.

\subsubsection{Methods}
For floating-point calculations we produce planted matrices $v^i = x^i \otimes y^i$, $i=1,\ldots,s$ where components $x^i$ and $y^i$ are generated as i.i.d.\ standard Gaussian random vectors in $\R^m$ and $\R^n$ respectively.
The additional non-planted matrices $v^i$, $s < i \leq R$ are generated as i.i.d.\ standard Gaussian matrices in $\R^{m \times n}$.

To perform a test, we used SVD to compute an orthonormal basis for the random subspace, which is used as input to the JLV algorithm.
An orthonormal basis is not required, but generally ensures that the rank-1 solutions are not basis elements and finding them cannot be done by inspecting individual basis elements.

As described in Proposition~\ref{prop:extract-v}, the JLV algorithm outputs a set of recovered rank-1 matrices, $\{\hat v^i\}_{i=1}^s$.
Our measure of success was a greedy matching error for the recovered planted matrices.
Treating the recovered rank-1 matrices $\{\hat v^i\}_{i=1}^s$ and ground truth rank-1 matrices $\{ v^i\}_{i=1}^s$ as vectors in $\R^{mn}$, perfect recovery would mean there is a matching via permutation $\pi:[s]\to [s]$ such that under the matching vectors are co-linear, i.e.,
\[
\frac{|\langle v^i, \hat{v}^{\pi(i)}\rangle|}{\|v^i\|_2 \|\hat{v}^{\pi(i)}\|_2} = 1, \quad \text{ for all } i\in [s].
\]
To measure deviations from perfect recovery, we perform a greedy matching to find $\pi$ and then report a worst-case matching error
\begin{equation}\label{eq:matching-err}
w = 1- \min_{i\in[s]} \frac{|\langle v^i, \hat{v}^{\pi(i)}\rangle|}{\|v^i\|_2 \|\hat{v}^{\pi(i)}\|_2},
\end{equation}
which gives the largest deviation from co-linearity.

We additionally wanted to understand the failure modes of the JLV algorithm, above the conjectured bound.
As noted after Proposition~\ref{prop:extract-v}, extracting the planted rank-1 matrices $\hat v^1, \ldots, \hat v^s$ is achieved using simultaneous diagonalization of a particular tensor.
When the bound~\eqref{eq:conj-cond} does not hold, there is no reason to believe that this tensor will be simultaneously diagonalizable but we wondered if the failure of \eqref{eq:goal} could have a ``perturbative'' effect on this tensor or if it would be catastrophic to the recovery of the planted rank-1 matrices.
To this end we modified the simultaneous diagonalization algorithm given in~\cite[Section 2.5]{JLV} so that in Step 3 when a set of eigenvalues cannot be matched, instead of outputting `Fail' we run a greedy matching procedure.
Similarly, we replace Step 4 of their algorithm, which requires  solving a linear system exactly, by a linear least squares method. These modifications result in a simultaneous diagonalization algorithm robust to small perturbations of its inputs.

\subsubsection{Results}
In the 6992 cases we tested where~\eqref{eq:conj-cond} holds, we ran the JLV algorithm on randomly generated subspaces and it successfully recovered the planted matrices with worst-case matching error (as defined in~\eqref{eq:matching-err}) no worse than $4.96 \times 10^{-14}$ (or correctly determined that there were no planted matrices in the case $s=0$).
These data and the programs to generate it are described in Section~\ref{sec:Code&Data}.
For the cases beyond where we have generated certificates using computations over finite fields, these computations give overwhelming evidence (Theorem*s~\ref{thm:numericalAll} and~\ref{thm:numericalParticular}) that Conjecture~\ref{Conjecture:SharpBoundary} still holds in these cases. These cases appear as cells colored purple and orange in Table~\ref{table:plant_max_values}.

In the 2481 cases we tested where~\eqref{eq:conj-cond} does not hold, we found that the JLV algorithm fails. For those cases with planted matrices ($s > 0$), the worst-case matching error was always at least 0.042, corresponding to an angular deviation of roughly $0.29$ radians or $16.6$ degrees. Furthermore, the \emph{average} value (over the 2481 cases) of the worst-case matching error was 0.759. To understand how the JLV algorithm fails, recall that JLV computes the subspace intersection $S^2(\U) \cap \SpEBasis$, and we know by Theorem~\ref{Thm:JLVFail} that when~\eqref{eq:conj-cond} fails, this subspace must contain ``spurious'' vectors in addition to the desired solutions $(v^i)^{\otimes 2}$. We saw that these spurious vectors had a catastrophic effect on the subsequent application of simultaneous diagonalization. Furthermore, we found that
\[
\dim(S^2(\U)\cap \SpEBasis) = \binom{R+1}{2} - \binom{m}{2}\binom{n}{2} \ge s+1
\]
for each case tested, meaning the bound~\eqref{eq:intersect_lowerbound} found in the proof of Theorem~\ref{Thm:JLVFail} held with \emph{equality}. For the cases where $s=0$, we know from Theorem~\ref{Thm:JLVFail} that $S^2(\U)\cap \SpEBasis$ must be non-empty, meaning JLV fails to certify, and our experiments confirmed this. See Section~\ref{sec:Code&Data} for further description of the data recorded.

\subsection{Extension to Symmetric Matrices}
\label{sec:GenToSymmMatrix}

Our work above considers subspaces $\U\subseteq \R^{m\times n}$ of matrices; however, it is also natural to consider subspaces of \emph{symmetric} matrices $\U \subseteq \Sym_m (\R) \subset \R^{m\times m}$, and to search for planted symmetric rank-1 matrices in such a subspace.
As one motivation, this problem is a subroutine in symmetric tensor decomposition, which arises when applying the method of moments.

First, we note that it is possible to define the space of symmetric rank-1 matrices in an identical manner to rank-1 matrices, as the algebraic variety cut out by the vanishing of all 2-by-2 minors, but where the ambient space is the space of symmetric matrices.
The variety is then
\[\X^\vee = \{v\in\Sym_m(\R): f(v)=0 \text{ for all } f\in P^\vee\} \subseteq \Sym_m(\R),\]
where we identify the variables $v_{ij}$ and $v_{ji}$ for all $i,j\in [m]$ to define
\[
P^\vee = \{v_{ab}v_{cd} - v_{ad}v_{cb} \, : \, 1\leq a \leq b \leq m, \, 1\leq c \leq d \leq m\}
\]
with $P^\vee \subseteq \R[v_{ij}:1\leq i\leq j \leq m]$.
Due to this identification, the set $P^\vee$ contains polynomials that are linearly dependent; the number of linearly independent polynomials for this variety is
\[
\binom{\binom{m+1}{2} + 1}{2} - \binom{m+3}{4} = \frac1{12} (m+1)m^2(m-1).
\]
Here, the first term is the dimension of the homogeneous degree-2 polynomials in variables $v_{ab}$, and the subtrahend is the dimension of $\Span\{v^{\otimes 2} : v\in \X^\vee\}$.

We construct a planted subspace as $\U = \Span\{v^1, \ldots, v^s, v^{s+1}, \ldots, v^R\} \subseteq \Sym_m(\R)$ where $v^1, \ldots, v^s \in \X^\vee$ are the \emph{planted symmetric matrices} and $v^{s+1}, \ldots, v^R \in \Sym_m(\R)$ are the \emph{non-planted symmetric matrices}.
The planted symmetric matrices may be expressed as $v^i = x^i(x^i)^\top$ for vectors $x^i\in\R^m$ when $i\leq s$.
Viewing the entries of the $x$ vectors as symbolic variables, there are $s\cdot m$ variables needed to express the planted symmetric matrices.
For the non-planted symmetric matrices, we take $v^i=z^i$ for symbolic $z$ variables with symmetry enforced: $z^i_{ab} = z^i_{ba}$ for all $a,b\in[m]$.
Accounting for symmetry, there are $(R-s)\binom{m+1}{2}$ $z$-variables in total. We concatenate the $x$ and $z$ variables into a single collection of symbolic variables, and assume these are chosen generically in the sense of Definition \ref{def:genericity}. We refer to
\[
\U = \Span\{x^1(x^1)^\top,\ldots,x^s(x^s)^\top,z^{s+1},\ldots,z^R\} \subseteq \Sym_m(\R)
\]
as a ``generically chosen planted subspace over the symmetric matrices with parameters $0 \le s \le R$.'' A generic subspace of this type is \emph{not} a generic instance of our original (non-symmetric) problem, so our prior methods are not immediately transferable to symmetric subspaces.

The JLV algorithm can be applied to the setting of symmetric matrices, and there is a natural analogue of \eqref{eq:goal} that implies its success, namely
\[S^2(\U)\ \cap\ \SpEBasis = \Span\{(v^1)^{\otimes 2},\ldots,(v^s)^{\otimes 2}\},
\label{eq:goal-vee} \tag{$\mathrm{Goal}^\vee$}\]
where $1\leq a < c \leq m$ and  $1\leq b < d \leq m$. From this, we can deduce a similar result to Theorem~\ref{Thm:JLVFail}, a negative result for subspaces over the symmetric matrices.

\begin{theorem} \label{Thm:JLVFail_symmetric}
If $\U = \Span\{v^1,\ldots,v^R\} \subseteq\R^{m\times m}$ is a generically chosen planted subspace over the symmetric matrices with parameters $0 \le s \le R$ where $\frac1{12} (m+1)m^2(m-1) < \binom{R+1}2 - s$, then \eqref{eq:goal-vee} fails to hold.
\end{theorem}

Since $s\leq R$, the condition above is implied by $\frac1{\sqrt6}m^2 + 1 < R$. The proof of Theorem~\ref{Thm:JLVFail_symmetric} is identical in form to the proof of Theorem~\ref{Thm:JLVFail} but replacing~\eqref{eq:goal} with~\eqref{eq:goal-vee} and the relevant dimensions with those appropriate to the symmetric case.
We can produce a similar bound
\begin{equation}
\label{eq:intersect_lowerbound_sym}
    \dim(S^2(\U)\cap \SpEBasis) \geq \binom{R+1}{2} - \frac1{12} (m+1)m^2(m-1).
\end{equation}

We now introduce a natural conjecture following similar principles to Conjecture~\ref{Conjecture:SharpBoundary}.
\begin{conjecture}
\label{Conjecture:SharpBoundarySymmetric}
If $\U\subseteq \Sym_m(\R)$ is a generically chosen planted subspace over the symmetric matrices with parameters $0 \le s \le R$, where $s$ and $R$ satisfy
\begin{equation}
\label{eq:conj-cond-sym}
    \frac1{12} (m+1)m^2(m-1) \ge \binom{R+1}{2} - s,
\end{equation}
then \eqref{eq:goal-vee} holds, and as a consequence the JLV algorithm succeeds.
\end{conjecture}

For large $m$, condition~\eqref{eq:conj-cond-sym} gives $R\sim \frac1{\sqrt6}m^2$.
To compare, Theorem~2 in~\cite{JLV} guarantees success of the JLV algorithm when $R \leq \frac16 m(m-1)$.
Hence, proof of Conjecture~\ref{Conjecture:SharpBoundarySymmetric} would indicate that the JLV algorithm succeeds for subspaces considerably larger than previously shown (by a factor of roughly 2.5).

While we do not have a proof of this conjecture, we have analogous results to Theorems~\ref{thm:certifiedAll} and~\ref{thm:certifiedParticular} that prove the conjecture in specific cases.
\begin{theorem}
    For all $2\leq m \leq 14$, Conjecture~\ref{Conjecture:SharpBoundarySymmetric} holds for every $s$ and $R$ satisfying inequality~\eqref{eq:conj-cond-sym}.
\end{theorem}
\begin{theorem}
    For all $2\leq m \leq 21$, Conjecture~\ref{Conjecture:SharpBoundarySymmetric} holds when $s=0$ or $s=R$ and satisfy inequality~\eqref{eq:conj-cond-sym}.
\end{theorem}

Similar to Definition~\ref{def:Rmax} we give the boundaries for subspace dimension depending on the proportion of planted rank-1 matrices.
\begin{definition}
The maximum subspace dimension $R$ that fits the conditions of Conjecture~\ref{Conjecture:SharpBoundarySymmetric} is
\[
    R^\vee_\mathrm{max}(m, s) := \left\lfloor\sqrt{\frac16(m+1)m^2(m-1)+\frac14 + 2s } - \frac12\right\rfloor.
\]
\end{definition}
Similarly, the maximum of $R^\vee_\text{max}(m, s)$ for any choice of $s$ has the following form.
\begin{definition}
    \[
        R^\vee_\mathrm{max}(m) := \left\lfloor\sqrt{\frac16(m+1)m^2(m-1)+\frac14} + \frac12\right\rfloor.
    \]
\end{definition}

\begin{table}
    \centering
    
    \begin{tabular}{c|*{13}{c}}
     $m$ & 2 & 3 & 4 & 5 & 6 & 7 & 8 & 9 & 10 & 11 & 12 & 13 & 14  \\
\hline
$R_\text{max}^\vee$ & \cellcolor{myblue}2 & \cellcolor{myblue}4 & \cellcolor{myblue}6 & \cellcolor{myblue}10 & \cellcolor{myblue}15 & \cellcolor{myblue}20 & \cellcolor{myblue}26 & \cellcolor{myblue}33 & \cellcolor{myblue}41 & \cellcolor{myblue}49 & \cellcolor{myblue}59 & \cellcolor{myblue}69 & \cellcolor{myblue}80  \\ \\ 

$m$ & 15 & 16 & 17 & 18 & 19 & 20 & 21 & 22 & 23 & 24 & 25 & 26 & 27 \\
\hline
$R_\text{max}^\vee$ & \cellcolor{mymagenta}92 & \cellcolor{mymagenta}104 &  \cellcolor{mymagenta}118 & \cellcolor{mymagenta}132 & \cellcolor{mymagenta}147 & \cellcolor{mymagenta}163 & \cellcolor{mymagenta}180 & 197 & 216 & 235 & 255 & 276 & 297
     \end{tabular}
     
    \caption{The value of $R^\vee_\mathrm{max}(m)$ for selected $m$. Cases for which we have proof of Conjecture~\ref{Conjecture:SharpBoundarySymmetric} for all $0\leq s \leq R$ are shown in \colorbox{myblue}{blue} and in \colorbox{mymagenta}{magenta} are cases for which the conjecture is proved only in the specific cases of $s=0$ and $s=R$, the latter having implications for symmetric tensor decomposition.}
    \label{table:symmetricProofAndNumericalResults}
\end{table}

\subsubsection{Certificates and Numerical Experiments}
We ran similar computations over a finite field to prove Conjecture~\ref{Conjecture:SharpBoundarySymmetric} for a range of matrix sizes.
The main difference was that planted matrices took the form $v^i = x^i\otimes x^i$ for $i\in [s]$.
The set $\{x^i\}_{i=1}^s$ was uniform on $(\F^m_p)^s$ conditional on being pairwise linearly independent for exact computations, and $\{x^i\}_{i=1}^s$ were independent standard Gaussian vectors in $\R^m$ for floating-point computations.
To produce additional basis elements from the space of symmetric matrices we selected entries in the upper triangular portion of each $v^i$, $i=s+1,\dots,R$ uniformly from $\F_p$ for exact computations, and i.i.d.\ Gaussian for floating-point computations.  We enforced the required symmetry by assigning entries below the diagonal to the value of their counterparts above the diagonal.

Findings were very similar to computations done in the previous sections. All results referred to here and additional data and programs to produce it are described in Section~\ref{sec:Code&Data}.
We produced certificates and proved Conjecture~\ref{Conjecture:SharpBoundarySymmetric} in all 441 cases we attempted (see Table~\ref{table:symmetricProofAndNumericalResults}).
Additional numerical tests verified that it was possible to recover the random planted symmetric rank-1 matrices via the JLV algorithm with the worst of the worst-case matching error being $1.89\times 10^{-15}$ (i.e.\ negligible error).

We also ran 414 tests over the bound~\eqref{eq:conj-cond-sym}, from $m=2$ to $m=14$ with parameters $0 \leq s \leq R^\vee_\mathrm{max}(m) + 1$ and $R = R^\vee_\mathrm{max}(m) + 1$ provided that $R \leq \binom{m+1}{2} - m$, i.e.\ dimension is one over the conjectured bound in~\eqref{eq:conj-cond-sym} but small enough that planted matrices will be identifiable.
Among these tests, the worst-case matching error was always at least 0.078 or an angular deviation of roughly 0.40 radians, indicating recovery was generally unsuccessful when we planted matrices ($s>0$). Furthermore, the \emph{average} value of the worst-case matching error was 0.776.
In addition, we found that
\[
\dim(S^2(\U)\cap \SpEBasis) = \binom{R+1}{2} - \frac1{12}(m+1)m^2(m-1) \ge s+1
\]
for every case tested, suggesting the bound in~\eqref{eq:intersect_lowerbound_sym} holds with equality for generic subspaces.

\subsection{Code and Data}
\label{sec:Code&Data}
Data and the programs written in Python can be found in the following \href{https://github.com/lemniscate8/linear-intersect-varieties}{repository}\footnote{Full URL is: https://github.com/lemniscate8/linear-intersect-varieties} hosted on GitHub.
Supplemental materials are organized as follows:
\begin{itemize}
    \item Python script \texttt{conjecture\_prover.py} was used to produce all certificates and numerical results. This program is supported by scripts,
    \begin{itemize}
        \item \texttt{linear\_conic\_intersect.py}, which contains a robust implementation of the simultaneous diagonalization algorithm among other subroutines;
        \item \texttt{subspace\_manipulation.py}, for performing symmetric lifts and Khatri--Rao product;
        \item and \texttt{segre\_veronese\_ideal\_generation.py}, which is used to generate an independent basis of polynomials that cut out rank-1 matrices and symmetric rank-1 matrices.
    \end{itemize}
       
    \item The \texttt{certificates} directory contains comma-separated value (CSV) files named according to the convention
    \texttt{<test\_type>(\_sym)\_b(<lower\_bound>-)<upper\_bound>\_p<modulus>.csv} that contain certificates for the partial proofs of Conjectures~\ref{Conjecture:SharpBoundary} and~\ref{Conjecture:SharpBoundarySymmetric}.
    \begin{itemize}
        \item \texttt{test\_type}'s of ``all'', ``null'' or ``cpd'' determine whether tests are performed with any number of planted solutions, no planted solutions, or number of planted solutions equal to subspace dimension, respectively.
        \item The modifier \texttt{\_sym} indicates the data set is for the symmetric matrix case.
        \item Data in each CSV is indexed by matrix shape (\texttt{m} and \texttt{n}), the maximum subspace dimension allowable, \texttt{R}, and number of planted solutions, \texttt{s}.
        Indices \texttt{m} and \texttt{n} in the file satisfy
        $
        \texttt{lower\_bound} \leq \frac1{\sqrt2}\texttt{mn}\leq \texttt{upper\_bound}
        $, or $
        \texttt{lower\_bound} \leq \frac1{\sqrt6}\texttt{m}^2\leq \texttt{upper\_bound}.
        $ in the symmetric case.
        The set of these indicies across all files gives all cases for which we have algorithmically proved Conjectures~\ref{Conjecture:SharpBoundary} and~\ref{Conjecture:SharpBoundarySymmetric}.
        \item Derived values:
        \begin{itemize}
            \item \texttt{seed}, a seed for psuedo-random number generation which along with index \texttt{m,n,R,s} can be used to produce a subspace spanned by vectors in $\Z$;
            \item \texttt{det\_mod<modulus>}, the value of a maximal minor of the $M$ matrix corresponding to the subspace modulo the integer \texttt{modulus};
            \item and \texttt{rm\_rows}, an array of non-negative integers for rows of the matrix which should be discarded to form the maximal square sub-matrix $M'$ whose determinant is computed.
        \end{itemize}
        \item We include the following CSV files generated by running \texttt{conjecture\_prover.py}:
        \begin{itemize}
            \item \texttt{certificates/all\_b60\_p997.csv},
            \item \texttt{certificates/all\_sym\_b90\_p997.csv},
            \item \texttt{certificates/cpd\_b61-120\_p997.csv},
            \item and \texttt{certificates/null\_b61-120\_p997.csv}.
        \end{itemize}
    \end{itemize}
    \item The \texttt{numerical} directory contains CSV files by the same naming convention that report our numerical tests running the full JLV algorithm.
    
    \begin{itemize}
        \item \texttt{test\_type} may also be ``overbound'' indicating that the subspace dimensions are one above the bound where \eqref{eq:goal} holds. In these cases we observe that the JLV algorithm fails.
        \item Indexing follows the same convention as in certificates.
        \item The value of \texttt{seed} is either drawn from the \text{certificates} data or set to zero if we have not produced a certificate for the index given.
        \item Derived values include
        \begin{itemize}
            \item \texttt{ker\_dim}, the dimension of $S^2(\U)\cap \SpEBasis$ estimated as the kernel of a particular matrix;
            \item \texttt{decomp\_error}, reports a diagnostic measure indicating how well robust simultaneous diagonalization succeeded;
            \item \texttt{s\_val}, reports the value of the smallest singular value of the $M$ matrix corresponding to this subspace;
            \item \texttt{w}, reports our worst case greedy matching error metric.
        \end{itemize}
        \item We include the following files generated by running \texttt{conjecture\_prover.py}:
        \begin{itemize}
            \item \texttt{numerical/all\_b75\_p997.csv},
            \item \texttt{numerical/all\_sym\_b90\_p997.csv},
            \item \texttt{numerical/cpd\_b76-165\_p997.csv},
            \item \texttt{numerical/cpd\_sym\_b91-190\_p997.csv},
            \item \texttt{numerical/null\_b76-165\_p997.csv},
            \item \texttt{numerical/null\_sym\_b91-190\_p997.csv},
            \item \texttt{numerical/overbound\_b60\_p997.csv},
            \item and \texttt{numerical/overbound\_sym\_b90\_p997.csv}.
        \end{itemize}
    \end{itemize}
    \item The \texttt{log} directory contains extra data from the certification and numerical procedures reporting how long tests took to run and if failures of one kind or another occurred.
\end{itemize}

\section{Applications to Tensor Decomposition}
\label{sec:Applications}

Verifying the success of the JLV algorithm beyond its previously proven bounds has immediate application for tensor decomposition, following the arguments in Section~6 of~\cite{JLV}.
For order-3 tensors, we have the following corollary to Theorem~\ref{Thm:GoalHolds}.

\begin{corollary}
\label{Cor:ThirdOrderDecomposition}
    The JLV algorithm can be used to decompose, in polynomial time, a generic tensor of rank $R$ in $\R^{n_1 \times n_2 \times n_3}$ when
    \[
    R \leq \min\left\{\frac12(n_1-1)(n_2-2), \, n_3\right\}.
    \]
\end{corollary}

\begin{proof}
We recast Proposition~25 (on page~30) in~\cite{JLV} in a context specific to order-3 tensor decomposition and use our improved bound.
Consider a tensor
\[
\mathcal T = \sum_{i=1}^R a^i\otimes b^i \otimes c^i \in  \R^{n_1 \times n_2 \times n_3}
\]
with generic components $a^i\in \R^{n_1}$, $b^i\in \R^{n_2}$, and $c^i\in \R^{n_3}$ for $i=1,\dots, R$.
Flattening $\mathcal T$ along its first two modes produces a matrix $T$ with rows indexed by pairs $(i,j)$ with $1\leq i \leq n_1$ and $1\leq j \leq n_2$ so that $T_{(i,j),k} := \mathcal T_{i,j,k}$.
Observe that $T$ has a factorization
\[
T = \begin{bmatrix}
    a^1 \otimes b^1 & a^2 \otimes b^2 & \dots & a^R \otimes b^R
\end{bmatrix}
\begin{bmatrix}
    (c^1)^\top \\
    (c^2)^\top \\
    \vdots \\
    (c^R)^\top
\end{bmatrix}.
\]
Since $R\leq n_3$, the vectors $\{c^i\}_{i=1}^R$ are generically linearly independent. 
This implies $\operatorname{colspan}(T) = \Span\{a^1 \otimes b^1, a^2 \otimes b^2, \dots, a^R \otimes b^R\}$.
By Corollary~\ref{Cor:JLVSuccess}, we can then recover the rank-1 matrices $\{a^i \otimes b^i \}_{i=1}^R$ up to scalar multiple and ordering from $\U  = \operatorname{colspan}(T)$.
By Lemma~\ref{lem:lin-indep}, the vectors $\{a^i \otimes b^i\}_{i=1}^R$ are generically linearly independent, which means we can solve a linear system for $\{c^i\}_{i=1}^R$ to recover the rank-1 terms $\{a^i\otimes b^i\otimes c^i\}_{i=1}^R$ up to ordering.
\end{proof}

Similar results apply for order-4 tensors.

\begin{corollary}
    \label{Cor:FourthOrderDecomposition} 
    The JLV algorithm can be used to decompose, in polynomial time, a generic tensor of rank $R$ in $\R^{n_1 \times n_2\times n_3\times n_4}$ when
    \[
    R\leq \min \left\{\frac12 (n_1-1)(n_2-2), \, n_2n_3\right\}.
    \]
\end{corollary}
For comparison, in the case $n_1=n_2=n_3=n_4=n$, our bound on $R$ is roughly twice the current bound of $R \le \frac{1}{4} (n-1)^2$, given by Corollary~28 of~\cite{JLV}. The proof of Corollary~\ref{Cor:FourthOrderDecomposition} is almost identical to that of Corollary \ref{Cor:ThirdOrderDecomposition}, where one flattens the tensor along the first two modes and uses the columns as a basis for the subspace $\U$.

If Conjecture~\ref{Conjecture:SharpBoundary} holds then the bounds above improve to $R\leq \{R_\text{max}(n_1,n_2), n_3\}$ for Corollary~\ref{Cor:ThirdOrderDecomposition} and $R\leq \{R_\text{max}(n_1,n_2), n_3n_4\}$ for Corollary~\ref{Cor:FourthOrderDecomposition}.

Our conjecture for the case of symmetric matrices also has implications for symmetric order-4 tensor decomposition.
\begin{corollary}
    \label{Cor:SymmetricFourthOrderDecomposition}
    Assuming Conjecture~\ref{Conjecture:SharpBoundarySymmetric}, the JLV algorithm can be used to decompose, in polynomial time, a generic symmetric tensor of rank $R$ in $\R^{n\times n\times n\times n}$ when
    \[
        R\leq \sqrt{\frac16n^2(n^2-1) +\frac14} + \frac12.
    \]
\end{corollary}
\begin{proof}
    This follows by similar steps as those presented in Corollary \ref{Cor:ThirdOrderDecomposition}.
    If the equality in \eqref{eq:goal-vee} holds under Conjecture~\ref{Conjecture:SharpBoundarySymmetric} then Proposition~25 in~\cite{JLV} now holds for a tensor with generic components when rank is $R \leq R^\vee_\mathrm{max}(m)$.
\end{proof}
For comparison, the current bound given by Corollary 30 in \cite{JLV} is $R \leq \frac16 n(n-1)$.
We have proven that our improved bound, which is $R\sim \frac1{\sqrt{6}}n^2$, holds for generic tensors up to $n=21$ with the values of $R$ for each $n$ given in Table~\ref{table:symmetricProofAndNumericalResults}.

\addcontentsline{toc}{section}{References}
\bibliography{main}
\bibliographystyle{alpha}

\end{document}